\pdfoutput=1
\documentclass{article}
\usepackage{amsmath}
\usepackage{amsthm}
\usepackage{mathptmx}
\usepackage[libertine]{newtxmath}
\usepackage[T1]{fontenc}
\usepackage{geometry}
\usepackage{color}
\usepackage{array}
\usepackage{verbatim}
\usepackage{pifont}
\usepackage{float}
\usepackage{graphicx}
\PassOptionsToPackage{normalem}{ulem}
\usepackage{ulem}
\usepackage{nameref}

\makeatletter

\providecommand{\tabularnewline}{\\}

\usepackage{etoolbox}
\usepackage{setspace}
\usepackage[bottom,hang]{footmisc}
\usepackage[medium,compact]{titlesec}	% must be before hyperref and redefining \(sub)*section! ; incompatible with LIPIcs
\usepackage{xcolor}
\usepackage{hyperref}	% must be after all footnote configuration packages except cleveref!
\hypersetup
{
  unicode=true,
  colorlinks=true,
  linkcolor=[rgb]{0 0 .8},
  citecolor=[rgb]{0 .8 0},
  urlcolor=[rgb]{.8 0 .8},
  breaklinks=true,  % uses \binoppenalty which is expected to be in [0..32767]
  pdfstartview=,
  bookmarksopen=true,
}
\usepackage[numbered]{bookmark}
\usepackage[all]{hypcap}
\usepackage[nameinlink]{cleveref}	% must be after hyperref!

\AtBeginDocument{\let\ref\cref}	% redefine after \maketitle instead if incompatible with some document class (such as old \documentclass[sigconf]{acmart})
\crefname{section}{Section}{Sections}
\crefname{figure}{Figure}{Figures}
\crefname{table}{Table}{Tables}
\crefname{example}{Example}{Examples}
\crefname{footnote}{}{}
\crefformat{footnote}{#2\footnotemark[#1]#3}
\makeatletter\g@addto@macro{\UrlBreaks}{\UrlOrds}\makeatother	% for latex->pdflatex ; adds [*-~'"-] to the list of URL breakpoints
\usepackage{graphicx}
\let\oldsqrt\sqrt
\renewcommand{\sqrt}[2][\ \,\,]{{\!\!\oldsqrt[\raisebox{.1em}{\scalebox{.7}{$#1$}}]{#2}\,}}
  \let\originalleft\left
  \let\originalright\right
  \renewcommand{\left}{\mathopen{}\mathclose\bgroup\originalleft}
  \renewcommand{\right}{\aftergroup\egroup\originalright}
\binoppenalty=\maxdimen
\DeclareMathSizes{10}{10}{8.5}{7}
\DeclareMathSizes{10.95}{10.95}{9}{7}	% for \documentclass[11pt]{article} ; The 10.95 is ridiculously dumb!
\DeclareMathSizes{12}{12}{10}{8}
\makeatletter
\AtBeginDocument
{
  \check@mathfonts
  \fontdimen13\textfont2 = \fontdimen13\textfont2		% display-style superscript height
  \fontdimen14\textfont2 = \fontdimen13\textfont2		% text-style non-fraction superscript height
  \fontdimen15\textfont2 = \fontdimen13\textfont2		% text-style fraction numerator/denominator superscript height
  \fontdimen16\textfont2 = 1.3\fontdimen16\textfont2	% subscript normal depth
  \fontdimen17\textfont2 = \fontdimen16\textfont2		% subscript changed depth if superscript too close
  \fontdimen14\scriptfont2 = \fontdimen13\scriptfont2	% \scriptfont is for script
  \fontdimen15\scriptfont2 = \fontdimen13\scriptfont2
  \fontdimen16\scriptfont2 = 1.3\fontdimen16\scriptfont2
  \fontdimen17\scriptfont2 = \fontdimen16\scriptfont2
  \fontdimen14\scriptscriptfont2 = \fontdimen13\scriptscriptfont2	% \scriptscriptfont is for script-script
  \fontdimen15\scriptscriptfont2 = \fontdimen13\scriptscriptfont2
  \fontdimen16\scriptscriptfont2 = 1.3\fontdimen16\scriptscriptfont2
  \fontdimen17\scriptscriptfont2 = \fontdimen16\scriptscriptfont2
}
\makeatother
\usepackage[skip=.5em]{caption}
\newcommand{\setmulength}[2]{#1=#2\relax}	% the calc (and mathtools) package makes \setlength incompatible with mu units
\setmulength{\thinmuskip}{2mu plus 1mu minus 1mu}
\setmulength{\medmuskip}{2mu plus 1mu minus 1mu}
\setmulength{\thickmuskip}{4mu plus 1mu minus 2mu}
\AtBeginEnvironment{pmatrix}
{
  \setlength{\arraycolsep}{.4em}
  
}
\usepackage{microtype}
\DisableLigatures{encoding=OT1,family=ntxtlf}
\DeclareSymbolFont{unspacedletters}{OT1}{ntxtlf}{m}{it}
\SetSymbolFont{unspacedletters}{bold}{OT1}{ntxtlf}{b}{it}
\DeclareMathSymbol{A}{\mathalpha}{unspacedletters}{`A}
\DeclareMathSymbol{B}{\mathalpha}{unspacedletters}{`B}
\DeclareMathSymbol{C}{\mathalpha}{unspacedletters}{`C}
\DeclareMathSymbol{D}{\mathalpha}{unspacedletters}{`D}
\DeclareMathSymbol{E}{\mathalpha}{unspacedletters}{`E}
\DeclareMathSymbol{F}{\mathalpha}{unspacedletters}{`F}
\DeclareMathSymbol{G}{\mathalpha}{unspacedletters}{`G}
\DeclareMathSymbol{H}{\mathalpha}{unspacedletters}{`H}
\DeclareMathSymbol{I}{\mathalpha}{unspacedletters}{`I}
\DeclareMathSymbol{J}{\mathalpha}{unspacedletters}{`J}
\DeclareMathSymbol{K}{\mathalpha}{unspacedletters}{`K}
\DeclareMathSymbol{L}{\mathalpha}{unspacedletters}{`L}
\DeclareMathSymbol{M}{\mathalpha}{unspacedletters}{`M}
\DeclareMathSymbol{N}{\mathalpha}{unspacedletters}{`N}
\DeclareMathSymbol{O}{\mathalpha}{unspacedletters}{`O}
\DeclareMathSymbol{P}{\mathalpha}{unspacedletters}{`P}
\DeclareMathSymbol{Q}{\mathalpha}{unspacedletters}{`Q}
\DeclareMathSymbol{R}{\mathalpha}{unspacedletters}{`R}
\DeclareMathSymbol{S}{\mathalpha}{unspacedletters}{`S}
\DeclareMathSymbol{T}{\mathalpha}{unspacedletters}{`T}
\DeclareMathSymbol{U}{\mathalpha}{unspacedletters}{`U}
\DeclareMathSymbol{V}{\mathalpha}{unspacedletters}{`V}
\DeclareMathSymbol{W}{\mathalpha}{unspacedletters}{`W}
\DeclareMathSymbol{X}{\mathalpha}{unspacedletters}{`X}
\DeclareMathSymbol{Y}{\mathalpha}{unspacedletters}{`Y}
\DeclareMathSymbol{Z}{\mathalpha}{unspacedletters}{`Z}
\DeclareMathSymbol{a}{\mathalpha}{unspacedletters}{`a}
\DeclareMathSymbol{b}{\mathalpha}{unspacedletters}{`b}
\DeclareMathSymbol{c}{\mathalpha}{unspacedletters}{`c}
\DeclareMathSymbol{d}{\mathalpha}{unspacedletters}{`d}
\DeclareMathSymbol{e}{\mathalpha}{unspacedletters}{`e}
\DeclareMathSymbol{f}{\mathalpha}{unspacedletters}{`f}
\DeclareMathSymbol{g}{\mathalpha}{unspacedletters}{`g}
\DeclareMathSymbol{h}{\mathalpha}{unspacedletters}{`h}
\DeclareMathSymbol{i}{\mathalpha}{unspacedletters}{`i}
\DeclareMathSymbol{j}{\mathalpha}{unspacedletters}{`j}
\DeclareMathSymbol{k}{\mathalpha}{unspacedletters}{`k}
\DeclareMathSymbol{l}{\mathalpha}{unspacedletters}{`l}
\DeclareMathSymbol{m}{\mathalpha}{unspacedletters}{`m}
\DeclareMathSymbol{n}{\mathalpha}{unspacedletters}{`n}
\DeclareMathSymbol{o}{\mathalpha}{unspacedletters}{`o}
\DeclareMathSymbol{p}{\mathalpha}{unspacedletters}{`p}
\DeclareMathSymbol{q}{\mathalpha}{unspacedletters}{`q}
\DeclareMathSymbol{r}{\mathalpha}{unspacedletters}{`r}
\DeclareMathSymbol{s}{\mathalpha}{unspacedletters}{`s}
\DeclareMathSymbol{t}{\mathalpha}{unspacedletters}{`t}
\DeclareMathSymbol{u}{\mathalpha}{unspacedletters}{`u}
\DeclareMathSymbol{v}{\mathalpha}{unspacedletters}{`v}
\DeclareMathSymbol{w}{\mathalpha}{unspacedletters}{`w}
\DeclareMathSymbol{x}{\mathalpha}{unspacedletters}{`x}
\DeclareMathSymbol{y}{\mathalpha}{unspacedletters}{`y}
\DeclareMathSymbol{z}{\mathalpha}{unspacedletters}{`z}
\pretocmd{\section}{\addvspace{0em plus 2em}\penalty-4000}{}{}
\pretocmd{\subsection}{\addvspace{0em plus 1.5em}\penalty-2000}{}{}
\pretocmd{\subsubsection}{\addvspace{0em plus 1em}\penalty-1000}{}{}
  \usepackage{enumitem}
  \setlistdepth{14}
  \makeatletter
  \renewcommand\p@enumii{\theenumi}
  \renewcommand\p@enumiii{\theenumi\theenumii}
  \renewcommand\p@enumiv{\theenumi\theenumii\theenumiii}
  \makeatother
  \newlength{\listspace}
  \setlist{topsep=\listspace,itemsep=\listspace,parsep=0em,partopsep=0em}
  \setlist[itemize]{leftmargin=1.5em,beginpenalty=4000,endpenalty=2000}
  \setlist[itemize,1]{leftmargin=1.5em,beginpenalty=2000,endpenalty=0}
  \setlist[enumerate]{leftmargin=2em,beginpenalty=4000,endpenalty=2000}
  \setlist[enumerate,1]{leftmargin=2em,beginpenalty=2000,endpenalty=0}
  \AtBeginEnvironment{itemize}{\interlinepenalty=4000}
  \AtBeginEnvironment{enumerate}{\interlinepenalty=4000}
\makeatletter\newcommand{\justified}{\rightskip\z@skip\leftskip\z@skip}\makeatother
\date{}
\usepackage{upgreek}
\usepackage{calc}	% makes \setlength incompatible with mu units!
\newlength{\linespace}
\newcommand\makelinespace{\setlength{\linespace}{\baselineskip-1em}\vspace{\linespace}}
\newlength{\parspace}
\newcommand\makeparspace{\setlength{\parspace}{\parskip+\baselineskip-1em}\vspace{\parspace}}
\newcommand{\forcefontspace}{\par\vspace{-\baselineskip}\vphantom{ABCDEgjpqy}}	% reserves space for ascenders/descenders %
\let\incgraphics\includegraphics
\newsavebox{\imagebox}
\newlength{\imagerule}
\newcommand{\imagescale}{1}
\newcommand{\scalegraphics}[1]{\renewcommand{\imagescale}{#1}}
\renewcommand{\includegraphics}[2][]
{%
\def\image{\scalebox{\imagescale}{\incgraphics[#1]{#2}}}%
\savebox{\imagebox}{\image}%
\setlength{\imagerule}{\ht\imagebox+\baselineskip-.7em}%
\ifvmode{\forcefontspace}\fi\rule[0em]{0em}{\imagerule}\image%
}
\let\oldfigure\figure
\let\oldtable\table
\makeatletter
\def\beginfloat{\centering\vspace{.1em}\makeparspace}
\def\figure@i[#1]{\oldfigure[#1]\beginfloat}
\def\figure@ii{\oldfigure\beginfloat}
\def\figure{\@ifnextchar[\figure@i \figure@ii}
\def\table@i[#1]{\oldtable[#1]\beginfloat}
\def\table@ii{\oldtable\beginfloat}
\def\table{\@ifnextchar[\table@i \table@ii}
\makeatother
\newcommand\beforefloat{\forcefontspace\vspace{.02em}}
\newcommand\afterfloat{\vspace{.02em}}
\BeforeBeginEnvironment{figure}{\beforefloat}
\BeforeBeginEnvironment{table}{\beforefloat}
\AfterEndEnvironment{figure}{\afterfloat}
\AfterEndEnvironment{table}{\afterfloat}
\newlength{\parskipcopy}
\makeatletter
\def\@minipagerestore{\setlength{\intextsep}{0em}\setlength{\parskip}{\parskipcopy}\vphantom{ABCDEgjpqy}\vspace{-\baselineskip}\vspace{-\parskip}}
\AtEndEnvironment{minipage}{\forcefontspace}
\AfterEndEnvironment{minipage}{\makelinespace}
\let\oldfbox\fbox
\renewcommand{\fbox}[1]{\vspace{.05em}\setlength{\fboxsep}{0em}\oldfbox{#1}\vspace{.2em}\ifvmode{\makelinespace\ensurelinespace}\fi}
\renewcommand{\boxed}[1]{\oldfbox{\m@th$#1$}}	% also makes \boxed not \displaystyle
\makeatother
\everymath{\colormath}
\newcommand{\colormath}{}
\newcommand{\setmathcolor}[1]{\renewcommand{\colormath}{\texorpdfstring{\color{#1}}{}}}
\newcommand{\resetmathcolor}{\renewcommand{\colormath}{}}
\newcommand\colorfnmark{}

\makeatletter
\let\oldfnmark\@makefnmark
\def\@makefnmark{{\everymath{\colorfnmark}\oldfnmark}}
\makeatother

\newtheoremstyle{lwq}
  {0em}	% top minimum spacing ; see \thmbegin
  {0em}	% bottom minimum spacing ; see \thmend
  {\normalfont}	% body font
  {0em}	% indent
  {\bfseries}	% head font
  {.}	% head punctuation
  {.3em plus .2em}	% head spacing
  {\thmname{#1}\thmnumber{ #2}\thmnote{ (#3)}}	% custom head specification
\newtheoremstyle{lwqprf}
  {0em}	% top minimum spacing ; see \prfbegin
  {0em}	% bottom minimum spacing ; see \prfend
  {\normalfont}	% body font
  {0em}	% indent
  {\itshape}	% head font
  {.}	% head punctuation
  {.3em plus .2em}	% head spacing
  {\thmname{#1}\thmnote{ (#3)}}	% custom head specification
\newlength{\thmspace}
\newlength{\prfspace}
\newcommand\thmbegin{\par\addvspace{\thmspace}\vspace{\parskip}\addpenalty{-300}}
\newcommand\prfbegin{\par\addvspace{\prfspace}\vspace{\parskip}}
\newcommand\thmend{\par\addvspace{\thmspace}\addpenalty{-150}}	% It is essential to have "\par" first for Lyx to work
\newcommand\prfend{\par\addvspace{\prfspace}\addpenalty{-200}}	% It is essential to have "\par" first for Lyx to work
\newcommand{\theoremname}{Theorem}
\theoremstyle{lwq}\newtheorem{thm}{\protect\theoremname}
\AtBeginEnvironment{thm}{\thmbegin} \AtEndEnvironment{thm}{\thmend}
\newcommand{\definitionname}{Definition}
\theoremstyle{lwq}\newtheorem{defn}[thm]{\protect\definitionname}
\AtBeginEnvironment{defn}{\thmbegin} \AtEndEnvironment{defn}{\thmend}
\newcommand{\setblockspace}
{
  \setlength{\topsep}{0em}
  \setlength{\itemsep}{\listspace}
  \setlength{\parsep}{0em}
  \setlength{\partopsep}{\listspace}
  \setlength{\leftmargin}{1.5em}
}
\newcommand{\setblockflush}{}
\makeatletter
\newcommand{\setblockpenalties}
  {\interlinepenalty=3000
  \ifdef{\insideblock}
    {\@beginparpenalty=4000\@endparpenalty=2000}
    {\@beginparpenalty=2000\@endparpenalty=0\def\insideblock{}}}
\makeatother
\newenvironment{block}
  {\par\nobreak\setblockpenalties\list{}{\setblockspace}\setstretch{1}\setblockflush}
  {\endlist}
\renewcommand{\proofname}{Proof}
\theoremstyle{lwqprf}\newtheorem{prf}{\protect\proofname}
\renewenvironment{proof}[1][]{\prfbegin\begin{prf}[#1]\pushQED{\qed}}{\popQED\end{prf}\prfend}
\renewcommand{\qed}{\hfill{}\hspace{2em minus 1em}\qedsymbol}
\newenvironment{centerbox}
{\par\begin{centering}}
{\par\end{centering}}
\newcommand{\lemmaname}{Lemma}
\theoremstyle{lwq}\newtheorem{lem}[thm]{\protect\lemmaname}
\AtBeginEnvironment{lem}{\thmbegin} \AtEndEnvironment{lem}{\thmend}
\newcommand{\remarkname}{Remark}
\theoremstyle{lwq}
\AtBeginEnvironment{rem}{\thmbegin} \AtEndEnvironment{rem}{\thmend}
\theoremstyle{lwq}\newtheorem*{rem*}{\protect\remarkname}
\AtBeginEnvironment{rem*}{\thmbegin} \AtEndEnvironment{rem*}{\thmend}
\usepackage[framemethod=TikZ]{mdframed}
\mdfdefinestyle{mdroundedboxinfloat}
{
  skipabove=0em,
  skipbelow=0em,
  roundcorner=.3em,
  innertopmargin=.3em,
  innerbottommargin=.3em,
  innerrightmargin=.3em,
  innerleftmargin=.3em,
  backgroundcolor=none,
  linecolor=gray,
}
\newenvironment{roundedboxinfloat}
{\vspace{-.3em}\begin{mdframed}[style=mdroundedboxinfloat]\vspace{.1em}\forcefontspace}
{\forcefontspace\end{mdframed}\unskip\vspace{-.24em}}

\AtBeginDocument{\fontsize{10.5}{13}\selectfont}
\AtBeginEnvironment{block}{\fontsize{10.5}{12}\selectfont}
\BeforeBeginEnvironment{tabular}{\makeparspace}
\AfterEndEnvironment{tabular}{\makeparspace}
\usepackage{url}
\definecolor{green}{rgb}{.3,.7,0}

\AtBeginDocument{

}

\makeatother

\begin{document}
\begin{comment}
Comments are used in this document, so do not delete this comment
so that the environment is defined.
\end{comment}

\renewcommand{\qed}{\hfill{}\hspace{2em minus 1em}\scalebox{1.4}{$\diamond$}}

\newcommand\br{\addpenalty{-1000}}

\newcommand\customref[2]{{\crefname{enumi}{#1}{}\ref{#2}}}
\newcommand\step[1]{\customref{step}}
\newcommand\point[1]{\customref{point}}
\makeatletter\newcommand\clabel[1]{\phantomsection\def\@currentlabelname{#1}}\makeatother

\global\long\def\nn{\mathbb{N}}
\global\long\def\zz{\mathbb{Z}}
\global\long\def\qq{\mathbb{Q}}
\global\long\def\rr{\mathbb{R}}

\global\long\def\wi{\subseteq}
\global\long\def\co{\supseteq}
\global\long\def\nwi{\nsubseteq}
\global\long\def\nco{\nsupseteq}
\global\long\def\none{\varnothing}
\global\long\def\less{\smallsetminus}

\global\long\def\ii{\mathbf{1}}
\global\long\def\pp{\mathbf{P}}
\global\long\def\ee{\mathbf{E}}
\global\long\def\vv{\mathbf{Var}}
\global\long\def\cv{\mathbf{Cov}}

\global\long\def\floor#1{\left\lfloor #1\right\rfloor }
\global\long\def\ceil#1{\left\lceil #1\right\rceil }

\global\long\def\cond#1#2#3{\left(\vphantom{#1#2#3}\right.\,#1\mathrel{\,?\,}\allowbreak#2\mathrel{\,:\,}\allowbreak#3\,\left.\vphantom{#1#2#3}\right)}

\global\long\def\f#1{\operatorname{#1}}

\global\long\def\a#1{.\mathrm{#1}}

\global\long\def\e{\upvarepsilon}

\global\long\def\tbox#1{\boxed{\text{#1}}}

\global\long\def\stbox#1{\scalebox{0.8}{\boxed{\text{#1}}}}

\global\long\def\smbox#1{\scalebox{0.8}{\boxed{#1}}}

\global\long\def\grey#1{\textcolor{gray}{#1}}

\global\long\def\ubrace#1#2{\underbrace{\underset{}{{\strut}#1{\strut}}}_{#2}}

\global\long\def\vsp#1{\vspace{#1em}}

\global\long\def\scale#1#2{\scalebox{#2}{\ensuremath{#1}}}

\global\long\def\lift#1#2{\raisebox{#2}{\ensuremath{#1}}}

\global\long\def\tilt#1#2{\rotatebox{#2}{\ensuremath{#1}}}

\global\long\def\tr{\mathbb{T}}

\scalegraphics{.7}

\newcommand\aligntop[1]
{
\savebox{\imagebox}{#1}
\raisebox{.5em-\ht\imagebox}{#1}
}
\newcommand\alignmid[1]
{
\savebox{\imagebox}{#1}
\raisebox{(\baselineskip-\ht\imagebox)/2}{#1}
}
\newcommand\overdown[1]
{
\savebox{\imagebox}{#1}
\smash{\raisebox{\baselineskip-\ht\imagebox}{#1}}
}
\newcommand\overmid[1]
{
\savebox{\imagebox}{#1}
\smash{\raisebox{(\baselineskip-\ht\imagebox)/2}{#1}}
}
\newcommand\overup[1]
{
\smash{#1}
}

\noindent \begin{center}
\textbf{\huge{}Optimal Multithreaded Batch-Parallel 2-3 Trees}
\par\end{center}{\huge \par}

\noindent \begin{center}
\begin{tabular}{>{\centering}p{0.3\textwidth}}
\textbf{\large{}Wei Quan Lim}{\large \par}

National University of Singapore\tabularnewline
\end{tabular}
\par\end{center}

\section*{Keywords}

Parallel data structures, pointer machine, multithreading, dictionaries,
2-3 trees.

\section*{Abstract}

This paper presents a \textbf{batch-parallel 2-3 tree $\tr$} in
an \textbf{\textit{asynchronous dynamic multithreading model}} that
supports searches, insertions and deletions in sorted batches and
has essentially optimal parallelism, even under the restrictive \textbf{\textit{QRMW
(queued read-modify-write) memory contention model}} where concurrent
accesses to the same memory location are queued and serviced one by
one.

Specifically, if $\tr$ has $n$ items, then performing an item-sorted
batch (given as a leaf-based balanced binary tree) of $b$ operations
on $\tr$ takes $O\left(b\cdot\log\left(\frac{n}{b}+1\right)+b\right)$
work and $O(\log b+\log n)$ span (in the worst case as $b,n\to\infty$).
This is information-theoretically \textbf{\textit{work-optimal}} for
$b\le n$, and also \textbf{\textit{span-optimal}} for pointer-based
structures. Moreover, it is easy to support \textbf{\textit{optimal}}
intersection, union and difference of instances of $\tr$ with sizes
$m\le n$, namely within $O\left(m\cdot\log\left(\frac{n}{m}+1\right)\right)$
work and $O(\log m+\log n)$ span. Furthermore, $\tr$ supports other
batch operations that make it a very useful building block for parallel
data structures.

To the author's knowledge, $\tr$ is the first \textbf{\textit{parallel
sorted-set data structure}} that can be used in an \textbf{\textit{asynchronous}}
multi-processor machine under a memory model with \textbf{\textit{queued
contention}} and yet have asymptotically optimal work and span. In
fact, $\tr$ is designed to have \textbf{\textit{bounded contention}}
and satisfy the claimed work and span bounds regardless of the execution
schedule.

Since all data structures and algorithms in this paper fit into the
dynamic multithreading paradigm, all their performance bounds are
directly \textbf{\textit{composable}} with those of other data structures
and algorithms in the same model. Finally, the \textbf{\textit{pipelining
techniques}} in this paper are also likely to be very useful in asynchronous
parallelization of other recursive data structures.

\section*{Acknowledgements}

I am exceedingly grateful to my family and friends for their unfailing
support, and to all others who have given me valuable comments and
advice. In particular, I would like to specially thank my brother
Wei Zhong Lim and my supervisor Seth Gilbert for very helpful discussions
and feedback. This research was partly supported by Singapore MOE
AcRF Tier 2 project MOE2018-T2-1-160.

\section{Introduction}

The \textbf{dynamic multithreading paradigm} (see~\cite{CormenLeRi09}
chap.~27) is a common parallel programming model underlying many
parallel languages and libraries such as Java~\cite{javaconcurrency},
OpenMP~\cite{OpenMP}, Cilk dialects~\cite{Cilk,IntelCilkPlus13},
Intel Thread Building Blocks~\cite{TBB} and the Microsoft Task Parallel
Library~\cite{TPL}. In this paradigm, programs can use programming
primitives such as threads, fork/join (also spawn/sync), parallel
loops and synchronization primitives, but cannot stipulate how the
subcomputations are scheduled for execution on the processors.

We consider a multithreaded procedure (which can be an algorithm or
a data structure operation) to be correct if and only if it has the
desired behaviour regardless of the execution schedule. Moreover,
we wish to obtain good bounds on the work and span of the procedure,
preferably independent of the execution schedule.

Unfortunately, many data structures and algorithms are designed in
theoretical computation models with synchronous processors, such as
the (synchronous) PRAM models, and so they can be difficult or impossible
to implement in dynamic multithreading with the same asymptotic work/time
bounds once we take memory contention into account, such as under
the QRMW (queued read-modify-write) contention model described in
\ref{sec:model}, which captures both the \textbf{\textit{asynchronocity}}
and \textbf{\textit{contention costs}} inherent in running multithreaded
procedures on most real multi-processor machines. Thus it is desirable
to have as many useful algorithms and data structures as possible
designed in computation models compatible with dynamic multithreading.

One indispensable data structure is the \textbf{map} (or \textbf{dictionary})
data structure, which supports searches/updates, inserts and deletes
(collectively referred to as \textbf{accesses}) of items from a linearly
ordered set. Balanced binary trees such as the AVL tree or the red-black
tree are commonly used to implement a sequential map, taking $O(\log n)$
worst-case cost (in the comparison model) per access for a tree with
$n$ items. A related data structure is the \textbf{sorted-set} data
structure, which supports intersection, union and difference of any
two sets. Using maps based on balanced binary trees to implement sorted-sets
yields $O(\min(m,n)\cdot\log\max(m,n))$ worst-case cost of each set
operation where $m,n$ are the sizes of the input sets.

The obvious question is whether we can have an efficient multithreaded
parallel map, or an efficient multithreaded sorted-set, or both. In
this paper we describe a pointer-based \textbf{multithreaded batch-parallel
2-3 tree} $\tr$ that is both information-theoretically \textbf{\textit{work-optimal}}
and \textbf{\textit{span-optimal}} even under the QRMW contention
model. Here the input batch is given as a leaf-based balanced binary
tree. Specifically, performing a sorted batch of $b$ accesses on
an instance of $\tr$ with $n$ items takes $O\left(b\cdot\log\left(\frac{n}{b}+1\right)+b\right)$
work and $O(\log b+\log n)$ span. This is superior to the work and
span bounds of the PVW 2-3 tree~\cite{paul1983paralleldict} despite
not having the luxury of lock-step synchronous processors.

Furthermore, since $\tr$ is a multithreaded data structure whose
performance bounds are independent of the schedule, it is trivially
\textbf{\textit{composable}}, which means that we can use $\tr$ as
a black-box data structure in any multithreaded algorithm and easily
obtain composable performance bounds. Indeed, the parallel working-set
map in \cite{OPWM} and the parallel finger structure in \cite{OPFS}
both rely such a parallel 2-3 tree as a key building block.

\section{Related Work}

To illustrate the difficulty of converting data structures designed
in the PRAM models to efficient multithreaded implementations, consider
the PVW 2-3 tree~\cite{paul1983paralleldict} that supports performing
an item-sorted batch of searches, insertions or deletions, which was
designed in the EREW PRAM model.

Performing a sorted batch of searches in the PVW 2-3 tree involves
splitting the batch into contiguous subbatches and pushing them down
the tree in non-overlapping waves. This is easy with synchronous processors,
but it is non-trivial to translate that pipelining technique to an
asynchronous setting with queued memory contention, since a naive
use of locking to prevent overlapping waves would cause the worst-case
span to increase from $O(\log b+\log n)$ to $O\left(\log b\cdot\log n\right)$.
\ref{sub:pipelined-splitting} shows how this can be done. But performing
a sorted batch of $b$ insertions or deletions on the PVW 2-3 tree
with $n$ items involves spawning $O(\log b)$ non-overlapping waves
of structural changes from the bottom of the 2-3 tree upwards to the
root, and for this there does not seem any way to eliminate the reliance
on the processors' lock-step synchronicity.

Other map data structures in the PRAM models include parallel B-trees
by Higham et al.~\cite{higham1991parallelbtree}, parallel red-black
trees by Park et al.~\cite{park2001parallelrbtree} and parallel
$(a,b)$-trees by Akhremtsev et al.~\cite{AkhremtsevS16}, all of
which crucially rely on lock-step synchronous processors as well.

A different approach of pipelining using futures by Blelloch et al.~\cite{BlellochRe97}
yields an implementation of insertion into a variant of PVW 2-3 trees
that requires not only a CREW/EREW PRAM but also a unit-time plus-scan
(all-prefix-sums) operation. The multithreaded parallel sorted-sets
presented by Blelloch et al. in \cite{BlellochFS16} take $O\left(b\cdot\left(\log\frac{n}{b}+1\right)\right)$
work but up to $\Theta(\log b\cdot\log n)$ span per operation between
two sets of sizes $n,b$ where $n\ge b$. The span was reduced to
$O(\log b+\log n)$ by Blelloch et al. in \cite{blelloch2019forkjoinalgo},
but that algorithm as written relies on $O(1)$-time concurrent reads
and may take $\Omega(\sqrt{n})$ span in a queued memory contention
model.

This paper shows that, using special pipelining schemes, it is actually
possible to design a \textbf{\textit{multithreaded}} batch-parallel
2-3 tree that takes $O\left(b\cdot\left(\log\frac{n}{b}+1\right)+b\right)$
work and $O(\log b+\log n)$ span, even if only \textbf{\textit{bounded
memory contention}} is permitted. The techniques shown here can likely
be adapted to pipeline top-down operations on many other tree-based
data structures.

\section{Main Results}

\label{sec:main}

This paper presents, to the author's best knowledge, the first \textbf{\textit{multithreaded
sorted-set data structure}} that can be run on an \textbf{\textit{asynchronous}}
parallel pointer machine and achieves \textbf{\textit{optimal}} work
and span bounds even under \textbf{\textit{queued memory contention}}.

Specifically, the underlying data structure $\tr$ is a pointer-based
batch-parallel 2-3 tree that works in the QRMW contention model (see
\ref{sec:model}) and supports performing an item-sorted batch of
$b$ accesses within $O\left(b\cdot\log\left(\frac{n}{b}+1\right)+b\right)$
work and $O(\log b+\log n)$ span (in the worst case as $b,n\to\infty$)
where $n$ is the current size of $\tr$ (\ref{sub:P23T-normal}).
Here we of course assume that we are given an $O(1)$-step comparison
function on pairs of items (i.e.~the comparison model), but there
is no loss of generality. This is information-theoretically \textbf{\textit{work-optimal}}
for $b\le n$, and also \textbf{\textit{span-optimal}} in the parallel
pointer machine model. Furthermore, the input batch can be any balanced
binary tree, including even another instance of $\tr$, and hence
$\tr$ can be used to implement \textbf{\textit{optimal}} persistent
sorted sets supporting intersection, union and difference of sets
with sizes $m\le n$ in $O\left(m\cdot\log\left(\frac{n}{m}+1\right)\right)$
work and $O(\log m+\log n)$ span (\ref{sec:PSS}).

$\tr$ also supports performing an unsorted batch of $b$ searches
within $O(b\cdot\log n)$ work and $O(\log b\cdot\log n)$ span (\ref{sub:P23T-unsorted}),
or an unsorted batch of $b$ accesses within $O(b\cdot\log\max(n,n'))$
work and $O\left((\log b)^{2}+\log n\right)$ span where $n'$ is
the size of $\tr$ after the batch operation (\ref{sub:P23T-full-access}).
These are useful when $b\gg n$. Additionally, $\tr$ supports performing
a reverse-indexing on an unsorted batch of $b$ direct pointers to
distinct items in it, which yields a sorted batch of those items within
$O\left(b\cdot\log\frac{n}{b}+b\right)$ work and $O(\log n)$ span
(\ref{sub:P23T-reverse}).

Actually, $\tr$ is designed to have \textbf{\textit{bounded contention}},
meaning that there is some constant $c$ such that every operation
on $\tr$ never makes more than $c$ concurrent accesses to the same
memory location.

\section{Key Ideas}

$\tr$ uses a \textbf{\textit{pipelined splitting scheme}} to partition
the 2-3 tree itself around the operations in the input batch, and
then performs each operation on its associated part, and then uses
a \textbf{\textit{pipelined joining scheme}} to join the parts back
up. Both pipelining schemes are top-down. The splitting scheme is
similar to the search in the PVW 2-3 tree, except that we push the
2-3 tree down the input batch, rather than the input batch down the
2-3 tree. But the joining scheme is completely different from the
bottom-up restructuring in the PVW 2-3 tree.

The main difficulty in both the splitting phase and joining phase
is in finding a top-down procedure that can be decomposed into `independent'
local procedures each of which runs in $O(1)$ span, which can then
be pipelined. This is not so hard for the splitting phase, but for
the joining phase this seems to require using integers to maintain
the structure of the spine nodes over a sequence of joins but without
actually performing the joins (see the outline in \ref{sec:P23T}).

\clearpage{}

\section{Parallel Computation Model}

\label{sec:model}

In this section, we describe the programming model as well as the
underlying memory model chosen in this paper.

\subsection{Programming Model}

\label{sub:prog-model}

\global\long\def\new#1{\f{new}\ \text{#1}}

We shall work within a high-level object-oriented dynamic multithreading
model that supports procedures and threads with standard multithreading
primitives (\textbf{\textit{terminate}}, \textbf{\textit{suspend}},
\textbf{\textit{fork}}, \textbf{\textit{join}}, \textbf{\textit{resume}}).
In this model, a thread $\tau$ can \textbf{\textit{terminate}} itself,
or \textbf{\textit{fork}} a new thread (obtaining a reference to it),
or \textbf{\textit{join}} to another thread $\upsilon$ (i.e.~wait
until $\upsilon$ terminates). Or $\tau$ can \textbf{\textit{suspend}}
itself (i.e.~temporarily stop running), and another thread (with
a reference to $\tau$) can \textbf{\textit{resume}} $\tau$ (i.e.~make
it continue running after the suspension). $\tau$ can also obtain
a reference to itself. Each of these takes $O(1)$ steps.

This programming model also supports the standard \textbf{RMW (read-modify-write)}
operations (including read, write, test-and-set, fetch-and-add, compare-and-swap),
as in almost all modern architectures. To capture contention costs,
we adopt the \textbf{QRMW (queued read-modify-write)} contention model,
as described in \cite{dwork1997contention}, in which RMW operations
on each memory location are FIFO-queued to be serviced, with only
one RMW operation on that memory location serviced per time step.
The thread making each memory request is blocked until the request
has been serviced. Additionally, we require each memory location to
be a named field of an object, denoted by ``$x\a d$'' where $x$
is a reference to the object and ``$\text{d}$'' is the name of
the field.

The actual complete execution of a multithreaded computation is captured
by its \textbf{execution DAG} $E$ (which may be schedule-dependent),
in which each node is a primitive instruction weighted by the time
taken to execute it, and the directed edges represent the computation
dependencies. Specifically, each thread $\tau$ executes a sequence
of instructions, the first one depending on the \textbf{\textit{fork}}
instruction that forked $\tau$, and every subsequent instruction
depending on the one just before it. The first instruction after a
\textbf{\textit{join}} instruction to join with thread $\upsilon$
depends also on the last instruction executed by $\upsilon$. And
the first instruction after a \textbf{\textit{suspend}} instruction
executed by $\tau$ depends also on the \textbf{\textit{resume}} instruction
that resumed $\tau$. But there are no dependencies between concurrent
accesses to the same memory location, even though they are linearized
during the actual execution under the QRMW contention model.

We can view $E$ as being dynamically generated as the computation
proceeds, where at any point during execution, a \textbf{ready node}
in $E$ (i.e.~node whose parents have been executed) corresponds
to a running thread, and a \textbf{scheduler} is used to assign ready
nodes to available processors (i.e.~processors that are not executing
any nodes) for execution. A \textbf{greedy scheduler} on each step
assigns as many unassigned \textbf{ready nodes} as possible to available
processors for execution.

We can now define work and span of a (terminating) multithreaded computation.
This allows us to capture the intrinsic costs incurred by the computation
itself, separate from the costs of any multithreaded program using
it.
\begin{defn}[Subcomputation Work/Span]
\label{def:work-span} Take any multithreaded computation on $p$
processors with execution DAG $E$, and any subcomputation $C$ (identified
with a subset of the nodes in $E$). The \textbf{work} taken by $C$
is the total weight of $C$. The \textbf{span} taken by $C$ is the
maximum possible total weight of the nodes in $C$ that lie on any
(directed) path in $E$.
\end{defn}
Note that work/span is \textbf{\textit{subadditive}} across subcomputations,
so performance bounds for algorithms and data structures in this model
are \textbf{\textit{composable}}. Moreover, all the algorithms and
data structures in this paper achieve the stated work/span bounds
independent of the scheduler.

In our analysis of a computation with execution DAG $E$, we will
often reason about procedure run \textbf{fragments}, meaning a contiguous
sequence of instructions executed during some procedure run. We will
also reason about the relations between events that happened during
the computation, where each such event $X$ is associated with some
set $\boldsymbol{\Gamma}_{X}$ of nodes in $E$ whose execution was
required for $X$ to happen, where $\Gamma_{X}$ is closed under ancestors.
Events are often described in terms of the start or end of certain
fragments, such as the start or end of a procedure run. For this we
define the following notion of span between events.
\begin{defn}[Inter-Event Span]
 Take any events $X,Y$ in a multithreaded computation with execution
DAG $E$. We say that $Y$ happens \textbf{within $s$ span after}
$X$ iff $\Gamma_{Y}\less\Gamma_{X}$ has at most $s$ span in $E$.
\end{defn}
For convenience, we shall also say that a fragment $C$\textbf{ }takes
$s$ \textbf{steps} iff $C$ has no \textbf{\textit{suspend}} node
and takes $s$ work, in which case $C$ also ends within at most $s$
span after it starts.

\subsection{Memory Model}

\label{sub:mem-model}

The QRMW contention model was chosen because the synchronous PRAM
model makes unrealistic assumptions including lock-step synchronicity
of processors and lack of collision on concurrent memory accesses
to the same locations \cite{gibbons1996asyncqrqw,gibbons1998qrqw,rauber2013parallel}.
For example, the load latency in the Cray XMT increases roughly linearly
with number of concurrent accesses to the same address but stays roughly
constant when the concurrent accesses are to random addresses \cite{secchi2011contentionxmt}.

In the QRMW PPM model, generalizing the PPM model in \cite{goodrich1996parallelsort}
to cater to the QRMW contention model, processors are asynchronous
and all accesses to shared memory are done via pointers, which can
be locally stored or tested for equality (but no pointer arithmetic).
More precisely, each pointer (if not $null$) is to a memory node,
which has a fixed number of memory cells. Each memory cell can hold
a single field, which is either an integer or a pointer. Each processor
also has a fixed number of local registers, each of which can hold
a single field. At each step, each processor (that has finished its
previous operation) can start any one of the following operations,
which except for RMW operations finishes in one step:
\begin{enumerate}
\item Perform a basic arithmetic operation~\footnote{In this paper we use only integer addition, subtraction, multiplication,
modulo and equality.} on integers in its registers, storing the result in another register.
\item Perform an equality-test between pointers in its registers, storing
the result ($0$ or $1$) in an integer register.
\item Perform an RMW operation on a memory cell via a pointer to the memory
node that it belongs to.
\item Create a new memory node, storing a pointer to it in a register.
\end{enumerate}
It is easy to see that the high-level programming model in \ref{sub:prog-model}
can be implemented in the QRMW PPM model if we have a greedy scheduler.
It also turns out that a greedy scheduler can be approximated in the
QRMW PPM model (with local RAM of size $p$) by a suitable work-stealing
scheduler such as the one in \cite{OPFS}, in the sense that any multithreaded
computation that takes $w$ work and $s$ span on $p$ processors
takes $O\left(\frac{w}{p}+s\right)$ expected time when run using
that work-stealing scheduler. All these results hold in the QRMW PRAM
model (i.e.~asynchronous PRAM with the QRMW contention model) as
well, though it is worth noting that the QRMW PPM model requires more
sophisticated techniques because we cannot use pointer arithmetic.

\section{Higher Synchronization Primitives}

\label{sec:high-sync}

In this model we can implement the following synchronization primitives:
\begin{enumerate}
\item \textbf{Non-blocking locks}: A lock $L$ can be \textbf{\textit{acquired}}
and \textbf{\textit{released}} over time, and is said to be \textbf{\textit{held}}
by a thread $\tau$ iff it has been acquired by $\tau$ but not yet
released. A non-blocking lock $L$ is implemented as a boolean field.
Operations on $L$ appear to be serialized (i.e.~queued and performed
one by one). $\f{\textbf{TryLock}}(L)$ attempts to acquire $L$ and
succeeds if $L$ is not currently held but fails otherwise, and returns
a boolean indicating whether it is successful. $\f{\textbf{Unlock}}(L)$
releases $L$. If at most $O(1)$ threads concurrently access $L$
(via any operation), then each access takes $O(1)$ steps (hence ``\textbf{\textit{non-blocking}}'').
\item \textbf{Barriers}: A single thread can \textbf{\textit{wait}} at a
barrier $B$ until another thread has \textbf{\textit{notified}} $B$
to allow any waiting thread to continue. Notification takes $O(1)$
steps, and waiting at $B$ takes $O(1)$ work and will return within
$O(1)$ span after the wait has started and $B$ has been notified.
There must be only one waiting thread and one notifying thread for
each barrier.
\item \textbf{Reactivation calls}: A procedure $P$ with no input/output
can be encapsulated by a reactivation wrapper, in which it can be
run only via \textbf{\textit{reactivations}}. If there are always
at most $O(1)$ concurrent reactivations of $P$, then the following
appear to hold (see \ref{thm:reactivation-prop} for more precise
guarantees):

\begin{enumerate}
\item Each reactivation call $C$ takes $O(1)$ steps and occurs at a single
point during $C$, called a \textbf{reactivation point}. All reactivation
points are distinct. At each reactivation point, if $P$ is not currently
running then it starts running (in another thread), otherwise it will
run again after its current run finishes.
\item Each run of $P$ starts within $O(1)$ span after \textbf{\textit{either}}
the end of the previous run of $P$ \textbf{\textit{or}} after the
start of a reactivation call with reactivation point after all previous
runs.
\item The total work done by the reactivation wrapper is $O(1)$ times the
number of reactivations.
\end{enumerate}
\end{enumerate}

All data structures and algorithms in this paper will be expressed
in procedural form using only the standard multithreading primitives,
these synchronization primitives and read/write (i.e.~we will not
directly use any other RMW operations).

We now explain how to implement the higher synchronization primitives
with proofs of their correctness and cost bounds.

\subsection{Non-Blocking Lock}

The \textbf{non-blocking lock} is trivially implemented using test-and-set
on a boolean field used to store the lock's state.
\begin{defn}[Non-Blocking Lock Operations]
\label{def:try-lock}~
\begin{block}
\item \textbf{TryLock( Bool Field $x$ ):}

\begin{block}
\item Return $\neg\f{TestAndSet}(x)$.
\end{block}
\item \textbf{Unlock( Bool Field $x$ ):}

\begin{block}
\item Set $x:=false$.
\end{block}
\end{block}
\end{defn}
If we wish, we can of course have a non-blocking lock \textbf{\textit{object}}
with a boolean field for its state, and implement the lock operations
on it in the same manner.

\subsection{Wait-Notify Barrier}

Using threads and test-and-set, we can also implement a \textbf{barrier},
where a condition $C$ is initially false, and one thread can wait
for $C$ to become true (getting blocked if it is not), and another
thread can notify that $C$ has become true.~\footnote{This is different from Java's semantics, because in Java a notify()
before a wait() would fail to trigger the waiting thread.}
\begin{defn}[Barrier]
\label{def:barrier}~
\begin{block}
\item Private Bool $C:=false$.
\item Private Bool $pass:=false$.
\item Private Thread $t:=null$.
\item \textbf{Public Wait():}

\begin{block}
\item Set $t:=\f{CurrentThread}$.
\item If $\neg\f{TestAndSet}(pass)$, suspend (current thread). 
\end{block}
\item \textbf{Public Notify():}

\begin{block}
\item Set $C:=true$.
\item If $\f{TestAndSet}(pass)$, resume $t$.
\end{block}
\item \textbf{Public Notified():}

\begin{block}
\item Return $C$.
\end{block}
\end{block}
\end{defn}
\begin{thm}[Barrier Properties]
\label{thm:barrier-prop} If there is only one call to Wait() and
one call to Notify(), then the following hold:
\begin{enumerate}
\item Wait() will not return before Notify() is called.
\item Wait() takes $O(1)$ work and will return within $O(1)$ span after
both Wait() and Notify() have been called.
\item Notify() takes $O(1)$ steps, after which Notified() will always return
$true$.
\end{enumerate}
\end{thm}
\begin{proof}
It is easy to see that expression ``$\f{TestAndSet}(pass)$'' is
evaluated exactly twice, once by Wait() and once by Notify(), the
first time to $false$ and the second time to $true$. If Wait() evaluates
that expression to $false$, then it suspends the thread calling it
after having stored (a pointer to) that thread in $t$, and so the
Notify() would subsequently evaluate that expression to $true$ and
resume that thread, within $O(1)$ span after Wait() and Notify()
have been called. If instead Wait() evaluates that expression to $true$,
Notify() must have already evaluated that expression to $false$,
and Wait() would return within $O(1)$ span.
\end{proof}

\subsection{Reactivation Wrapper}

Next is the \textbf{reactivation} wrapper for an inputless procedure
$P$, which guarantees roughly that the runs of $P$ will never overlap
if $P$ is run only via reactivations (i.e.~by calling Reactivate()),
and that there is always a complete run of $P$ as soon as possible
after each reactivation of $P$, and moreover each reactivation triggers
at most one run of $P$.
\begin{defn}[Reactivation Wrapper]
\label{def:reactivation}~
\begin{block}
\item Private Procedure $P$.\quad{}// $P$ is the procedure to be guarded
by the wrapper.
\item Private Int $count:=0$.
\item \textbf{Public Reactivate():}

\begin{block}
\item If $\f{FetchAndAdd}(count,1)=0$:

\begin{block}
\item Fork:

\begin{block}
\item Do:

\begin{block}
\item Set $count:=1$.
\item Call $P()$.
\end{block}
\item While $\f{FetchAndAdd}(count,-1)>1$.
\end{block}
\end{block}
\end{block}
\end{block}
\end{defn}
We now give the precise guarantees of the reactivation wrapper along
with an illustrated example.
\begin{centerbox}
\begin{figure}[H]
\begin{roundedboxinfloat}
\begin{centerbox}
\setmathcolor{brown}

\global\long\def\reactpoint{\lift{\downarrow}{.3em}}

\global\long\def\initpoint{\lift{\tilt{\Lsh}{180}}{.7em}}

\begin{tabular}{lcccccccccccccccccccccc}
\cline{3-6} \cline{8-22} 
Reactivation & \multicolumn{1}{c|}{} & $C_{1}$ &  &  & \multicolumn{1}{c|}{} & \multicolumn{1}{c|}{} & $C_{5}$ &  &  &  &  &  &  &  &  &  &  &  &  &  & \multicolumn{1}{c|}{} & \tabularnewline
\cline{3-6} \cline{8-22} 
calls: &  & $\boldsymbol{\reactpoint}$ &  &  &  &  &  &  &  &  &  &  &  &  & $\boldsymbol{\reactpoint}$ &  &  &  &  &  &  & \tabularnewline
\cline{5-5} \cline{7-8} \cline{11-11} \cline{15-17} 
 &  &  & \multicolumn{1}{c|}{} & \multicolumn{1}{c|}{$C_{2}$} & \multicolumn{1}{c|}{} & $C_{3}$ & \multicolumn{1}{c|}{} &  & \multicolumn{1}{c|}{} & \multicolumn{1}{c|}{$C_{4}$} &  &  & \multicolumn{1}{c|}{} & $C_{6}$ &  & \multicolumn{1}{c|}{} &  &  &  &  &  & \tabularnewline
\cline{5-5} \cline{7-8} \cline{11-11} \cline{15-17} 
 &  &  &  & $\boldsymbol{\reactpoint}$ &  &  & $\reactpoint$ &  &  & $\boldsymbol{\reactpoint}$ &  &  &  &  &  & $\reactpoint$ &  &  &  &  &  & \tabularnewline
$\vsp{-.5}$ &  &  &  &  &  &  &  &  &  &  &  &  &  &  &  &  &  &  &  &  &  & \tabularnewline
\cline{5-8} \cline{10-12} \cline{14-15} \cline{19-21} 
Runs of $P$: &  &  & \multicolumn{1}{c|}{$\initpoint$} & $R_{1}$ &  &  & \multicolumn{1}{c|}{} & \multicolumn{1}{c|}{$\initpoint$} & $R_{2}$ &  & \multicolumn{1}{c|}{} & \multicolumn{1}{c|}{$\initpoint\ $} & $R_{3}$ & \multicolumn{1}{c|}{} &  &  & \multicolumn{1}{c|}{$\initpoint$} & $R_{4}$ &  & \multicolumn{1}{c|}{} &  & \tabularnewline
\cline{5-8} \cline{10-12} \cline{14-15} \cline{19-21} 
\end{tabular}

$\vsp{.5}$

\resetmathcolor

\textit{\small{}Each rectangle represents the time interval for a
reactivation call $C_{i}$ or a run $R_{i}$ of $P$.}\\
\textit{\small{}``$\reactpoint$'' represents a reactivation point,
and ``$\initpoint$'' represents an initiation point.}\\
\textit{\small{}$R_{1},R_{2},R_{3,}R_{4}$ are triggered by $C_{1},C_{2},C_{4},C_{5}$
respectively.}\\
\textit{\small{}$R_{1}$ must start within $O(1)$ span after the
start of $C_{1}$.}\\
\textit{\small{}$R_{2},R_{3}$ must each start within $O(1)$ span
after the end of the previous run.}\\
\textit{\small{}$R_{4}$ starts within $O(1)$ span after the previous
run or after the start of $C_{5}$.}{\small \par}
\end{centerbox}
\caption{\label{fig:reactivation} Example subcomputation generated by the
reactivation wrapper for $P$}
\end{roundedboxinfloat}
\end{figure}
\end{centerbox}
\begin{thm}[Reactivation Wrapper Properties]
\label{thm:reactivation-prop} Consider only the reactivation calls
to $P$ (i.e.~calls to Reactivate()), and only the runs of $P$ initiated
by them (i.e.~via calls of $P()$ from Reactivate()). Then the following
hold if $P$ is reactivated by at most $k$ threads concurrently (see
\ref{fig:reactivation}):
\begin{enumerate}
\item No two (such) runs of $P$ overlap (in time).
\item Each reactivation call $C$ takes $O(k)$ span.
\item Each reactivation call $C$ can be associated with some \textbf{reactivation
point} during $C$, and each run $R$ of $P$ can be associated with
some \textbf{initiation point} before $R$ but after any previous
run of $P$, such that the following hold:

\begin{enumerate}
\item For every reactivation call with reactivation point $t$, there is
some run of $P$ with initiation point after $t$.
\item For every run $R$ of $P$ with initiation point $t$, there is some
reactivation call $C$ with reactivation point $s$ before $t$, such
that $s$ is the earliest reactivation point after all initiation
points before $t$, and we say that $C$ is the reactivation call
that \textbf{triggers} $R$. Furthermore, $R$ starts within $O(k)$
span after \textbf{\textit{either}} the end of the previous run of
$P$ \textbf{\textit{or}} after the start of $C$, and in the latter
case $s$ is after the end of any previous run of $P$.
\end{enumerate}
\item The reactivation wrapper does $O(k)$ work per reactivation call.
\end{enumerate}
\end{thm}

\begin{proof}
At any time, we say that a thread is at a line of code iff the thread
is currently executing that line (not yet finished) or will execute
that line next. For the sake of this analysis, we consider the execution
of a line involving a $\f{FetchAndAdd}$ operation to be finished
once the $\f{FetchAndAdd}$ operation itself has modified $count$
(i.e.~after the $\f{FetchAndAdd}$ operation has modified $count$,
the thread will immediately be at the appropriate next line). And
we say that a thread is in the \textbf{critical section} iff it is
at some line after ``If $\f{FetchAndAdd}(count,1)=0$:''.

We shall establish the invariant that $count\ge0$ and $count\ne0$
iff there is a thread in the critical section, and that there is at
most one such thread. Note that any thread that executes the ``Fork:''
line exits the critical section immediately after that, at the same
time as the forked thread enters the critical section at the ``Do:''
line. Note also that a thread that enters or exits the critical section
does so immediately after executing a line involving ``$\f{FetchAndAdd}(count,1)$''
or ``$\f{FetchAndAdd}(count,-1)$''. Thus it suffices to show that
the invariant is preserved whenever $count$ is changed (i.e.~a thread
finishes executing some line involving $count$):
\begin{itemize}
\item Whenever $count$ is changed from zero: By the invariant, there was
no thread in the critical section before that change, Thus it must
be due to some thread $\tau$ executing ``$\f{FetchAndAdd}(count,1)$''
resulting in $count$ being changed to $1$, upon which $\tau$ immediately
enters the critical section, preserving the invariant.
\item Whenever $count$ is changed from nonzero: By the invariant, before
that change $count>0$ and there was a unique thread $\tau$ in the
critical section.

\begin{itemize}
\item If the change was due to $\tau$ executing ``$\f{FetchAndAdd}(count,-1)$'':
After that change, either $count=0$ and $\tau$ exits the critical
section, or $count>0$ and $\tau$ remains in the critical section,
in either case preserving the invariant.
\item If the change was due to another thread executing ``$\f{FetchAndAdd}(count,1)$'':
That thread does not enter the critical section, and after that change
$count>0$ still, and $\tau$ is still in the critical section, preserving
the invariant.
\item If the change was due to $\tau$ executing ``$count:=1$'', then
after the change we have that $count=1$ and $\tau$ is still in the
critical section, preserving the invariant.
\end{itemize}
\end{itemize}

We can now use the invariant to prove the desired properties. Note
that there are at most $k$ concurrent evaluations of ``$\f{FetchAndAdd}(count,1)$''
and by the invariant at most one thread in the critical section modifying
$count$, and so every access to $count$ takes $O(k)$ work/span. 

Property~1 is satisfied, since the only calls to $P$ are from the
critical section. Property~2 is also satisfied, since it accesses
$count$ once and performs at most one fork.

We now turn to Property~3. Define the \textbf{reactivation point}
for a reactivation call $C$ to be when it finishes evaluating ``$\f{FetchAndAdd}(count,1)$'',
and the \textbf{initiation point} for a run $R$ of $P$ due to executing
``Call $P()$.'' to be when that instruction is reached (i.e.~at
the end of the execution of the previous instruction ``$count:=1$'').
Clearly the reactivation point for $C$ is during $C$, and the initiation
point for $R$ is before $R$ but after any previous run of $P$.

To verify Property~3a, consider any thread $\tau$ that makes a reactivation
call $C$ and finishes evaluating ``$\f{FetchAndAdd}(count,1)$''
at time $t$. The result of that evaluation satisfies $r\ge0$ by
the invariant, so there are two cases:
\begin{itemize}
\item If $r=0$: $\tau$ enters the critical section and reaches ``Call
$P()$.'' after time $t$.
\item If $r>0$: By the invariant, there is a unique thread $\sigma$ in
the critical section just before time $t$, and $count>1$ at time
$t$. If $\sigma$ was at a line before ``Call $P()$.'' at time
$t$, then it reaches ``Call $P()$.'' later. If not, then $\sigma$
was at a line after ``$count:=1$'' at time $t$, so no thread can
decrease $count$ after time $t$ until $\sigma$ evaluates ``$\f{FetchAndAdd}(count,-1)$'',
upon which it remains in the critical section and reaches ``Call
$P()$.'' again later.
\end{itemize}
To verify Property~3b, note that the reactivation call $C_{0}$ with
the earliest reactivation point evaluates ``$\f{FetchAndAdd}(count,-1)$''
to $0$ and the thread that executes $C_{0}$ is the first to enter
the critical section, so the first run $R_{0}$ of $P$ starts within
$O(k)$ span after $C_{0}$ starts, and hence Property~3b holds for
$R_{0}$. Thus it suffices to verify Property~3b for any run $R$
that has a previous run $R'$. Let $t,t'$ be the initiation points
for $R,R'$ respectively, and let $\tau$ be the thread that executed
$R'$. Since $count=1$ at both time $t'$ and time $t$, and $\tau$
decreases $count$ via ``$\f{FetchAndAdd}(count,-1)$'' at some
time $u$ between $t'$ and $t$, and by the invariant $count$ can
only be increased by an evaluation of ``$\f{FetchAndAdd}(count,1)$'',
there must be some reactivation point between $t'$ and $t$. Let
$C$ be the reactivation call with the earliest reactivation point
$s$ between $t'$ and $t$. There are two cases:
\begin{itemize}
\item If $s$ is before $u$: $count>1$ just before $u$, so $\tau$ remains
in the critical section at time $u$, and hence $R$ starts within
$O(k)$ span after $R'$.
\item If $s$ is after $u$: $s$ is after the end of $R'$, and $C$ evaluates
``$\f{FetchAndAdd}(count,1)$'' to $0$ at time $s$, so $R$ starts
within $O(k)$ span after the start of $C$.
\end{itemize}
Property~4 follows from Property~3b, since each reactivation call
(excluding the forked subcomputations) takes at most $O(k)$ work,
and the total work done by the resulting forked subcomputations excluding
the work done by the calls to $P$ is at most $O(k)$ times the number
of runs of $P$.
\end{proof}

\clearpage{}

\section{Basic Parallel Batch Operations}

In this section we shall show how to do some basic operations on batches
in the QRMW PPM model:
\begin{itemize}
\item \textbf{Filter} a batch of $n$ items based on an $O(1)$-time condition,
within $O(n)$ work and $O(\log n)$ span.
\item \textbf{Partition} a sorted batch of $n$ items around a sorted batch
of $k$ pivots, within $O\left(k\cdot\log\left(\frac{n}{k}+1\right)+k\right)$
work and $O(\log n+\log k)$ span.
\item \textbf{Join} a batch of $b$ batches of items, with $n$ items in
total, within $O(b+n)$ work and $O(\log b+\log n)$ span.
\item \textbf{Merge} two sorted batches of items, with $n$ items in total,
within $O(n)$ work and $O(\log n)$ span.
\item \textbf{Sort} a batch of $n$ items within $O(n\cdot\log n)$ work
and $O\left((\log n)^{2}\right)$ span.
\end{itemize}
We will always store any (non-empty) \textbf{batch} of items in a
\textbf{BBT}, namely a leaf-based height-balanced binary tree (i.e.~with
the items only at its leaves). Each binary tree $T$ is identified
with its root node $T\a{\textbf{root}}$, and each node $v$ of $T$
stores the following:
\begin{itemize}
\item $v\a{\textbf{left}}$ and $v\a{\textbf{right}}$ are its left and
right child nodes respectively.
\item $v\a{\textbf{height}}$ and $v\a{\textbf{size}}$ are the height and
number of leaves respectively of the subtree at $v$.
\item $v\a{\textbf{first}}$ and $v\a{\textbf{last}}$ are the first item
and last item respective in the subtree at $v$.
\end{itemize}
Note that the item at each leaf $v$ of $T$ is stored in $v\a{first}=v\a{last}$.
For convenience, we shall also use $\boldsymbol{V}(T)$ to denote
the nodes of $T$, and $\boldsymbol{L}(T)$ to denote the leaves of
$T$, and $\boldsymbol{H}(T,h)$ to denote the nodes of $T$ with
subtree height $h$ (i.e.~$v\in H(T,h)$ iff $v\a{height}=h$). Also,
for any leaves $v,w$ of $T$ we shall write ``$v\boldsymbol{\preceq}_{T}w$''
to mean that $v$ is before or equal to $w$ in $T$, and we shall
drop the subscript if it is clear from the context.

Each node of the BBT for a batch has additional fields that store
various values and pointers to objects needed for the parallelization
technique, which will be specified in the description of each parallel
batch operation.

\subsection{Pipelined Splitting}

\label{sub:pipelined-splitting}

The \textbf{pipelined splitting scheme} is the key technique employed
here for these parallel batch operations. This scheme is used to solve
the problem of efficiently distributing the leaves of a binary tree
$T$ to the leaves of another binary tree $U$ in a weak-order-preserving
manner, meaning that each leaf $v$ of $T$ is sent to some leaf $f(v)$
of $U$, and for every leaves $v,w$ of $T$ such that $v\preceq_{T}w$
we have $f(v)\preceq_{U}f(w)$. This scheme can be used as long as
we can always determine within $O(1)$ span whether $f(v)\preceq_{U}w$
given any leaf $v$ of $T$ and any leaf $w$ of $U$.

The basic idea is that we can \textbf{push $T$ down} $U$ in a \textbf{\textit{pipelined}}
fashion: We first push the entire tree $T$ to the root of $U$, and
whenever a subtree $B$ of $T$ arrives at an internal node $v$ of
$U$, we push the whole $B$ down to $v\a{left}$ if $f(B\a{last})\preceq v\a{left}\a{last}$,
or to $v\a{right}$ if $f(B\a{first})\npreceq v\a{left}\a{last}$,
but if neither holds then we \textbf{split} $B$, each time pushing
the appropriate half (i.e.~left or right subtree) down to a child
of $v$ and continuing to split the other half. Note that the subtrees
that arrive at each node of $U$ form a slice of $T$ (\ref{def:bt-slice}),
and at most one of them will be split.
\begin{defn}[Binary Tree Slice]
\label{def:bt-slice} A \textbf{slice} of a binary tree $T$ is a
sequence of disjoint non-sibling subtrees of $T$ that contain a set
of consecutive leaves of $T$. An \textbf{ordered slice} of $T$ is
a slice of $T$ that has the subtrees listed in rightward order in
$T$.
\end{defn}
It turns out that we can use queues $v\a{queue}[1]$ and $v\a{queue}[2]$
to store the unprocessed subtrees of $T$ at each node $v$ of $U$,
and maintain the \textbf{splitting invariant} that $\f{reverse}(v\a{queue}[1])+v\a{queue}[2]$
forms an ordered slice of $T$. To do so, when we process a subtree
of $T$ from $v\a{queue}[i]$, if we push it down whole to a child
$w$ of $v$ then we push it onto $w\a{queue}[i]$, otherwise if we
are splitting it then we always push the split subtrees onto $v\a{left}\a{queue}[2]$
or $v\a{right}\a{queue}[1]$. \ref{fig:slice-queues} illustrates
this.
\begin{centerbox}
\begin{figure}[H]
\begin{centerbox}
\begin{roundedboxinfloat}
\begin{centerbox}~

\includegraphics[bb=0bp 0bp 140bp 115bp,clip]{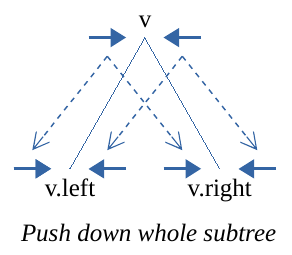}\quad{}\quad{}\includegraphics[bb=0bp 0bp 140bp 115bp,clip]{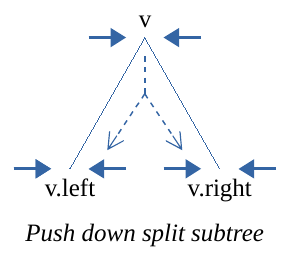}

\textit{\small{}``$\rightarrow v\leftarrow$'' represents node $v$
with $v\a{queue}[1]$ on its left and $v\a{queue}[2]$ on its right.}{\small \par}

\textit{\small{}Dotted arrows indicate which queue a subtree is pushed
to.}{\small \par}
\end{centerbox}
\caption{\label{fig:slice-queues} Splitting scheme to maintain the invariant
that the queued subtrees at each node form an ordered slice}
\end{roundedboxinfloat}
\end{centerbox}
\end{figure}

\end{centerbox}
This scheme can be carried out using a procedure $v\a{pushdown}[i]$
for each $v\a{queue}[i]$ that is run only via reactivation calls
(see \ref{sec:high-sync}) and is reactivated whenever a subtree is
pushed onto $v\a{queue}[i]$. We start by pushing $T$ onto $U\a{root}\a{queue}[1]$.
Each run of $v\a{pushdown}[i]$ processes one subtree $B$ from $v\a{queue}[i]$
(if any) and then reactivates itself if $v\a{queue}[i]$ was non-empty.
Processing $B$ means either pushing it down whole or forking a separate
thread to split it (repeatedly pushing down the appropriate half and
continuing to split the other half). Since at most one subtree that
arrives at $v$ is split, the splitting invariant guarantees that
every subtree that arrives at $v$ after that will only be pushed
onto the outer queues $v\a{left}\a{queue}[1]$ or $v\a{right}\a{queue}[2]$.
Therefore no concurrent pushes or concurrent pops are ever performed
on any queue.

Thus each queue can be implemented using a \textbf{dedicated queue},
which is a wait-free single-producer single-consumer queue implemented
by a linked list $L$ in the following manner. $L$ maintains pointers
to both the first node $L\a{head}$ and the last node $L\a{tail}$,
and every node $v$ in $L$ except $L\a{tail}$ stores an item $v\a{value}$
and a pointer to the next node $v\a{next}$, and $L\a{tail}\a{next}=null$.
Initially $L\a{head}=L\a{tail}$. The dedicated queue operations are
implemented as follows:
\begin{itemize}
\item \textbf{Push( $x$ ):} Create Node $w$ with $w\a{next}:=null$. Set
$L\a{tail}\a{value}:=x$. Set $L\a{tail}\a{next}:=w$. Set $L\a{tail}=w$.
\item \textbf{Pop():} Set $h:=L\a{head}$. If $h\a{next}\ne null$, set
$L\a{head}:=h\a{next}$. Return $h\a{value}$.
\end{itemize}
It can be shown that the total span is just $O(T\a{height}+U\a{height})$
(see \ref{thm:par-filter-cost}). But we must know what exactly the
pipelined splitting scheme is used for to get a good bound on the
total work.

\subsection{Parallel Filtering}

\label{sub:par-filter}

Parallel filtering an (unsorted) batch $T$ according to a condition
$C$, without changing the order in the batch, is done in 3 phases:
\begin{enumerate}
\item \textbf{\uline{Preprocessing phase}}\textbf{:} Each item in $T$
that satisfies $C$ has a rank in the sublist of $T$ that satisfies
$C$, which we shall call its filtered-rank. Recursively compute for
each node $v$ in $T$ the number $v\a{count}$ of filtered items
(i.e.~items that satisfy $C$) in the subtree at $v$, as well as
the range $v\a{range}$ of filtered-ranks of the filtered items in
the subtree at $v$. Then construct a blank batch $U$ of size $T\a{root}\a{count}$
that is a complete BBT (i.e.~every level is full except perhaps the
last), and compute for each node $w$ of $U$ the number $w\a{count}$
of leaves in its subtree and the range $w\a{range}$ of their ranks
in $U$. And place a barrier $w\a{done}$ at each leaf $w$ of $U$.
\item \textbf{\uline{Push-down phase}}\textbf{:} Use the pipelined splitting
scheme (\ref{sub:pipelined-splitting}) to push $U$ down $T$, where
a subtree $B$ of $U$ is pushed down whole to a node $v$ of $T$
iff $B\a{range}\wi v\a{range}$. Then clearly each leaf of $U$ will
be pushed down to a unique leaf of $T$ that has an item satisfying
$C$, and the order of those leaves in $U$ is the same as the order
of those items in $T$. Thus when a leaf $w$ of $U$ reaches a leaf
$v$ of $T$, we can simply copy the item from $v$ to $w$ and then
notify $w\a{done}$.
\item \textbf{\uline{Collating phase}}\textbf{:} After initiating the
push-down phase, wait on $w\a{done}$ for each leaf $w$ of $U$,
before returning $U$. Then clearly $U$ is only returned after the
push-down phase has finished.
\end{enumerate}
We shall now give the technical details, including the specific push-down
phase obtained by applying the pipelined-splitting scheme here, and
the splitting invariant involved.
\begin{defn}[Parallel Filtering]
 Parallel filtering an (unsorted) batch (or more generally a leaf-based
binary tree) $T$ according to a condition $C$, without changing
the order of items in the batch, is done via the following procedure:
\begin{block}
\item First \textbf{preprocess} the input batch and prepare the output batch
$U$:

\begin{enumerate}
\item Recursively for each node $v$ of $T$, compute the number $v\a{count}$
of items in the subtree at $v$ that satisfy $C$. Then recursively
compute $v\a{rstart}$ and $v\a{rend}$ defined by $T\a{root}\a{rstart}=0$
and $v\a{left}\a{rstart}=v\a{rstart}$ and $v\a{right}\a{rstart}=v\a{rstart}+v\a{left}\a{count}$
and $v\a{rend}=v\a{rstart}+v\a{count}$ for each internal node $v$
of $T$. 
\item If $T\a{root}\a{count}=0$, return a blank output batch (skipping
all the other phases).
\item Construct a blank complete BBT $U$ of size $T\a{root}\a{count}$,
and compute for each node $v$ of $U$ the number $v\a{count}$ of
leaves in its subtree, and $v\a{rstart}$ and $v\a{rend}$ defined
in exactly the same way as for $T$. 
\item In parallel place at each leaf $v$ of $U$ a Barrier $v\a{done}$.\quad{}//
see \ref{def:barrier}
\end{enumerate}
\item Then \textbf{push $U$ down $T$} to the appropriate leaf nodes via
a \textbf{pipelined splitting scheme}:

\begin{enumerate}
\item In parallel place at each node $v$ of $T$ two newly created (empty)
DedicatedQueues $v\a{queue}[1..2]$.
\item Define \textbf{feeding} $B$ to $v\a{queue}[i]$ to be pushing $B$
onto $v\a{queue}[i]$ and then reactivating $v\a{pushdown}[i]$.
\item Start by feeding $U\a{root}$ to $T\a{root}\a{queue}[1]$.
\item Whenever $v\a{pushdown}[i]$ is reactivated for some node $v$ of
$T$, it does the following:

\begin{block}
\item Pop subtree $B$ off $v\a{queue}[i]$. If $B=null$ (i.e.~$v\a{queue}[i]$
was empty), return.
\item Reactivate $v\a{pushdown}[i]$.
\item If $v$ is a leaf, copy the item from $v$ into $B$ and then call
$B\a{done}\a{notify()}$ and return.
\item If $B\a{rend}\le v\a{left}\a{rend}$, feed $B$ to $v\a{left}\a{queue}[i]$
and return.
\item If $B\a{rstart}\ge v\a{right}\a{rstart}$, feed $B$ to $v\a{right}\a{queue}[i]$
and return.
\item \clabel{$\ast$}\label{enu:splitter} Fork the following \textbf{splitting
process}:\hfill{}($\ast$)

\begin{block}
\item While $B$ is not a leaf:

\begin{block}
\item If $B\a{left}\a{rend}\le v\a{left}\a{rend}$:

\begin{block}
\item Feed $B\a{left}$ to $v\a{left}\a{queue}[2]$ and set $B:=B\a{right}$.
\end{block}
\item Otherwise:

\begin{block}
\item Feed $B\a{right}$ to $v\a{right}\a{queue}[1]$ and set $B:=B\a{left}$.
\end{block}
\end{block}
\item If $B\a{rend}\le v\a{left}\a{rend}$, feed $B$ to $v\a{left}\a{queue}[2]$,
otherwise feed $B$ to $v\a{right}\a{queue}[1]$.
\end{block}
\end{block}
\end{enumerate}
\item And (after starting the push-down phase) \textbf{collate} the results
by the following steps:

\begin{enumerate}
\item Call $v\a{done}\a{wait}()$ for each leaf $v$ of $U$ in parallel
(and wait for all to finish).
\item Recursively update $v\a{first}$ and $v\a{last}$ for each node $v$
of $U$.
\item Return $U$.
\end{enumerate}
\end{block}
\end{defn}

In our subsequent analysis we shall frequently use the fact that runs
of $v\a{pushdown}[i]$ do not overlap (by \ref{thm:reactivation-prop}
Property~1), and so we shall not mention it.
\begin{lem}[Parallel Filtering Invariants]
\label{lem:par-filter-inv} \nameref{sub:par-filter} satisfies the
following for each node $v$ of $T$:
\begin{enumerate}
\item The subtrees fed to $v$ (i.e.~to either $v\a{queue}[1]$ or $v\a{queue}[2]$)
form a slice of $U$.
\item The subtrees fed to $v\a{queue}[1]$ are (strictly) on the left in
$U$ of those fed to $v\a{queue}[2]$.
\item The subtrees fed to $v\a{queue}[1]$ are in (strictly) leftward order
and increasing depth in $U$.
\item The subtrees fed to $v\a{queue}[2]$ are in (strictly) rightward order
and increasing depth in $U$.
\item At most one subtree fed to $v$ is split, after which only the splitting
process (\nameref{enu:splitter}) will ever push onto $v\a{left}\a{queue}[2]$
or $v\a{right}\a{queue}[1]$.
\end{enumerate}
\end{lem}
\begin{proof}
We use structural induction on $T$. Invariants~1,2 follow from themselves
for $v\a{parent}$ (i.e.~the parent of $v$). Invariants~3,4 follow
from themselves and Invariants~2,5 for $v\a{parent}$. To establish
Invariant~5, observe that by Invariant~1 at most one subtree fed
to $v$ is not pushed down whole, and if it was a subtree $B$ from
$v\a{queue}[1]$ then the following hold at that point:
\begin{itemize}
\item Every subtree from $v\a{queue}[2]$ is on the right of $B$, by Invariant~2,
and so none of it gets fed to $v\a{left}\a{queue}[2]$. Hence from
then on only the splitting process feeds to $v\a{left}\a{queue}[2]$.
\item Every subsequent subtree from $v\a{queue}[1]$ is on the left of $B$,
by Invariant~3, and so none of it gets fed to $v\a{right}\a{queue}[1]$.
Also, every preceding subtree from $v\a{queue}[1]$ had already been
fed to $v\a{right}\a{queue}[1]$ in a preceding run of $v\a{pushdown}[1]$.
Hence from then on only the splitting process feeds to $v\a{right}\a{queue}[1]$.
\end{itemize}
Likewise if it was a subtree $B$ from $v\a{queue}[2]$, by symmetry.
Therefore Invariant~5 holds.
\end{proof}
Using these invariants we can prove both the correctness and costs
of parallel filtering. In particular, by Invariant~5 there are no
concurrent pushes performed on the dedicated queues, as required,
and hence the pipelined splitting scheme runs correctly. And now we
shall bound the parallel filtering costs.
\begin{lem}[Tree Path Length Sum]
\label{lem:tree-path-sum} Given any tree $T$ with $e$ edges, for
each node $v$ of $T$ let $m(v)$ be the length of the shortest path
from $v$ to a leaf if $v$ has at least $2$ children, but $0$ otherwise.
Then $\sum_{v\in V(T)}m(v)\le e$.\end{lem}
\begin{proof}
Let $C$ be the nodes of $T$ with at least $2$ children. For each
node $v$ of $T$, let $P(v)$ be the downward path from $v$ that
first takes the leftmost edge and then takes rightmost edges all the
way to a leaf. Observe that $P(v)$ and $P(w)$ have disjoint edges
for any distinct $v,w$ in $C$, since if they pass through a common
node $x$ then it must be that one of them starts at $x$, so that
path takes the leftmost edge, which is not the rightmost edge taken
by the other path. Therefore $\sum_{v\in V(T)}m(v)\le\sum_{v\in C}P(v)\a{length}\le e$.\end{proof}
\begin{defn}[Log-Splitting Property]
\label{def:log-splitting} We say that a binary tree $T$ is \textbf{$c$-log-splitting}
if every slice of $T$ containing $k$ leaves of $T$ has at most
$c\cdot\log_{2}(k+1)$ subtrees of $T$. Note that every BBT is $4$-log-splitting.
\end{defn}

\begin{thm}[Parallel Filtering Costs]
\label{thm:par-filter-cost} Parallel filtering a batch $T$ of size
$n$ according to a condition $C$ takes $O(n\cdot w)$ work and $O(\log n+s)$
span if every evaluation of $C$ takes $O(w)$ work and $O(s)$ span.\end{thm}
\begin{proof}
The preprocessing phase clearly takes $O(n\cdot w)$ work and $O(\log n+s)$
span. And the collating phase clearly takes $O(n)$ work and ends
within $O(\log n)$ span after it has started and the push-down phase
has ended. So it only remains to show that the push-down phase takes
$O(n)$ work and $O(\log n)$ span. Clearly initializing the nodes
of $T$ takes $O(n)$ work and $O(T\a{height})\wi O(\log n)$ span.

Now observe that the work taken by the push-down phase is $O(1)$
times the number of feedings, since every self-reactivation of $v\a{pushdown}[i]$
corresponds to a unique subtree that had been popped off $v\a{queue}[1]$
or $v\a{queue}[2]$, and every while-loop iteration in the splitting
process at $v$ feeds a subtree to a child of $v$. Each node $v$
of $T$ has $O(\log(v\a{count}+1))\wi O(\log v\a{size}+1)$ subtrees
fed to it (by \ref{lem:par-filter-inv} Invariant~1 since $U$ is
$4$-log-splitting), and $\log v\a{size}$ is $O(1)$ times the length
of the shortest path from $v$ to a leaf (since $T$ is a BBT). Thus
the number of feedings is $O(1)$ times the number of edges of $T$
(by \ref{lem:tree-path-sum}), which is $O(n)$.

To bound the span taken by the push-down phase, we set forth some
convenient definitions. Call a run of $v\a{pushdown}[i]$ a \textbf{$(v,i)$-run},
and a reactivation of $v\a{pushdown}[i]$ a \textbf{$(v,i)$-reactivation}.
Call a $(v,i)$-run \textbf{effective} iff it process some subtree
$X$ (i.e.~it pops $X$ off $v\a{queue}[i]$ and $X\ne null$). Call
a $(v,i)$-run a \textbf{$(v,i,X)$-run} iff it processes the subtree
$X$. Partition each run of the splitting process (\nameref{enu:splitter})
into fragments around the while-loop boundaries, and call each such
fragment a \textbf{$(v,X)$-splitter} iff it splits the subtree $X$
(i.e.~$B=X$ at the start of the loop). And for each subtree $X$
of $U$ let $d(X)$ be its depth in $U$, and observe that $d(X')<d(X)$
for any $(v,i,X')$-run that precedes a $(v,i,X)$-run (by \ref{lem:par-filter-inv}
Invariants~3,4).

Note that each $(v,i)$-run performs a $(v,i)$-reactivation iff it
is effective. Also, there cannot be three consecutive ineffective
$(v,i)$-runs $R_{1},R_{2},R_{3}$ in that order, otherwise (by \ref{thm:reactivation-prop}
Property~3b) there are distinct $(v,i)$-reactivations $C_{1},C_{2}$
in that order such that $C_{1}$ triggers $R_{2}$ and $C_{2}$ triggers
$R_{3}$, so $C_{1}$ ends after all $(v,i)$-runs preceding $R_{1}$
and $C_{2}$ starts before $R_{3}$, and hence $C_{1}$ and $C_{2}$
must be executed by feedings to $v$ and cannot overlap, which implies
that some subtree $X$ is fed to $v\a{queue}[i]$ between $C_{1}$
and $C_{2}$, but then $X$ must be processed by one of $R_{1},R_{2},R_{3}$,
contradicting their ineffectiveness.

Now consider any $(v,i)$-run $R$. Let $X$ be the last subtree processed
by the $(v,i)$-runs up to $R$ (which exists since the first $(v,i)$-run
processes the first tree fed to $v\a{queue}[i]$). There are two cases
(by \ref{thm:reactivation-prop} Property~3b):
\begin{itemize}
\item $R$ starts within $O(1)$ span after the end of the previous $(v,i)$-run
$R'$. Note that the last subtree $X'$ processed by the $(v,i)$-runs
up to $R'$ satisfies $d(X')\le d(X)$.
\item $R$ starts within $O(1)$ span after the start of the $(v,i)$-reactivation
$C$ that triggers $R$, and the reactivation point for $C$ is after
the end of any previous $(v,i)$-run. In this case, the previous $(v,i)$-run
$R'$ (if any) is ineffective, and hence $C$ must be executed by
the feeding of some subtree $X'$ to $v\a{queue}[i]$. Observe that
any feeding of a subtree $X''$ to $v\a{queue}[i]$ before $X'$ must
push $X''$ onto $v\a{queue}[i]$ before executing some $(v,i)$-reactivation
$C'$ before $C$, and $C'$ must have reactivation point before the
start of $R'$ (since $C$ triggers $R$), and hence $X''$ must have
been processed by some $(v,i)$-run preceding $R$ since $R'$ is
ineffective. Thus $R$ would process $X'$ if it had not already been
processed by an earlier $(v,i)$-run, and hence $d(X')\le d(X)$.
\end{itemize}
From this and the fact that there cannot be three consecutive ineffective
$(v,i)$-runs, we can deduce that for every $(v,i,X)$-run $R$, at
least one of the following holds:
\begin{itemize}
\item $R$ starts within $O(1)$ span after the end of some preceding $(v,i,X')$-run
where $d(X')<d(X)$.
\item $R$ starts within $O(1)$ span after the start of some $(v,i)$-reactivation
executed by the feeding of some subtree $X'$ to $v\a{queue}[i]$
where $d(X')\le d(X)$.
\end{itemize}
Moreover, except for the first feeding at the start of the push-down
phase, every feeding of a subtree $X$ to $v\a{queue}[i]$ is executed
by either some $(v\a{parent},i,X)$-run or some $(v\a{parent},X)$-splitter.
And every $(v,X)$-splitter starts within $O(1)$ span after either
the start of some $(v,i,X)$-run or the end of some $(v,X')$-splitter
where $d(X')<d(X)$.

Therefore by induction every $(v,i,X)$-run $R$ starts within $O(k+m+1)$
span after the push-down phase starts, where $k$ is the depth of
$v$ in $T$ and $m$ is the depth of $X$ in $U$. Thus the whole
push-down phase finishes within $O(T\a{height}+U\a{height}+1)\wi O(\log n)$
span.
\end{proof}
Sometimes, it is also useful to use parallel filtering on a leaf-based
binary tree that may not be a BBT but is sufficiently balanced. One
suitable notion is as follows.
\begin{defn}[$c$-balanced Binary Trees]
 We say that a binary tree $T$ is \textbf{$c$-balanced} (where
$c\in\rr^{+}$) iff $v\a{height}\le c\cdot\log_{2}(v\a{size})$ for
every node $v$ of $T$.\end{defn}
\begin{rem*}
Every BBT is $2$-balanced. Every red-black tree is also $2$-balanced.
\end{rem*}
We now prove a combinatorial lemma for $c$-balanced binary trees,
which we will then use to prove the cost bounds for performing parallel
filtering on such trees.
\begin{lem}[Balanced Tree Bound]
\label{lem:bal-tree-bound} $\sum_{v\in V(T)}\log_{2}v\a{size}\le4c\cdot n$
for every $c$-balanced binary tree $T$ with $n$ leaves.\end{lem}
\begin{proof}
For each leaf $x$ of $T$, the $d$-th ancestor $v$ of $x$ satisfies
$v\a{size}\ge2^{v\a{height}/c}\ge2^{d/c}>3$ if $d\ge2c$, since $T$
is $c$-balanced. Let $A(v)$ be the ancestors of any node $v$ of
$T$ (including $v$). Then $\sum_{v\in V(T)}\log_{2}v\a{size}=\sum_{x\in L(T)}\sum_{v\in A(x)}\frac{\log_{2}v\a{size}}{v\a{size}}$
$<\sum_{x\in L(T)}\left(\sum_{d=0}^{2c-1}\frac{\log_{2}3}{3}+\sum_{d=2c}^{\infty}\frac{d/c}{2^{d/c}}\right)$
because $\frac{\log_{2}k}{k}\le\frac{\log_{2}3}{3}$ for every $k\in\nn^{+}$
and $\frac{\log_{2}k}{k}\le\frac{\log_{2}m}{m}$ for every reals $k\ge m>3$.
Since we have $\sum_{d=0}^{2c-1}\frac{\log_{2}3}{3}\le2c$ and $\sum_{d=2c}^{\infty}\frac{d/c}{2^{d/c}}=\sum_{k=2}^{\infty}\sum_{i=0}^{c-1}\frac{(k\cdot c+i)/c}{2^{(k\cdot c+i)/c}}$
$<\sum_{k=2}^{\infty}\left(c\cdot\frac{k+1}{2^{k}}\right)=2c$, the
desired claim follows.\end{proof}
\begin{thm}[General Parallel Filtering Costs]
\label{thm:par-filter-cost-gen} Parallel filtering an $O(1)$-balanced
leaf-based binary tree $T$ of size $n$ and height $h$ according
to a condition $C$ takes $O(n\cdot w)$ work and $O(h+s)$ span if
every evaluation of $C$ takes $O(w)$ work and $O(s)$ span.\end{thm}
\begin{proof}
To show the work bound, we just need to bound the number of feedings,
which is $O\left(\sum_{v\in V(T)}\left(\log v\a{size}+1\right)\right)\wi O(n)$
by \ref{lem:bal-tree-bound}. The proof of the span bound is the same
as before.
\end{proof}
Parallel filtering can also be used to efficiently change the shape
of the underlying BBT of any batch, such as to a complete binary tree.
\begin{defn}[Parallel Balancing]
 Parallel balancing a batch $T$ of items to make the underlying
BBT be a complete binary tree, without changing the order in the batch,
can be done by parallel filtering (\ref{sub:par-filter}) with no
condition (i.e.~the condition $C$ always returns $true$).\end{defn}
\begin{rem*}
In general, we can obtain any desired shape of the underlying BBT
of $T$, by simply constructing the output batch to have the desired
shape.\end{rem*}
\begin{thm}[Parallel Balancing Costs]
 Parallel balancing a batch of $n$ items takes $O(n)$ work and
$O(\log n)$ span.\end{thm}
\begin{proof}
The claim follows immediately from \nameref{thm:par-filter-cost}
(\ref{thm:par-filter-cost}).\end{proof}
\begin{rem*}
Since a complete binary tree can be easily transformed into a 2-3
tree represented by a red-black tree, parallel balancing can be used
to transform any batch of $n$ items into a 2-3 tree within $O(n)$
work and $O(\log n)$ span.
\end{rem*}

\subsection{Parallel Partitioning}

\label{sub:par-part}

Exactly the same technique allows us to do parallel multi-way partitioning
of a sorted batch $T$ of items around a sorted batch $P$ of pivot
items, in 3 similar phases:
\begin{enumerate}
\item \textbf{\uline{Preprocessing phase}}\textbf{:} Insert $\infty$
into $P$ (treating $\infty$ as more than every item).
\item \textbf{\uline{Push-down phase}}\textbf{:} Use the pipelined splitting
scheme (\ref{sub:pipelined-splitting}) to push $T$ down $P$, where
a subtree $B$ of $T$ is pushed down whole from a node $v$ of $P$
to $v\a{left}$ iff $B\a{last}\le v\a{left}\a{last}$ and to $v\a{right}$
iff $B\a{first}>v\a{left}\a{last}$. Then clearly each item $x$ in
$T$ will be pushed down (in some subtree of $T$) to the leftmost
leaf $v$ of $P$ such that $x\le v\a{last}$. We also \textbf{close}
$v$ when no more subtrees will be fed to $v$. Once $v$ is closed
and both $v\a{queue}[1]$ and $v\a{queue}[2]$ are empty, we wait
for the splitting process at $v$ to finish before closing both $v\a{left}$
and $v\a{right}$.
\item \textbf{\uline{Collating phase}}\textbf{:} After starting the push-down
phase, for each leaf $v$ of $P$, wait for $v$ to be closed before
joining the subtrees in each $v\a{queue}[i]$ and then tagging $v$
with the join of the results.
\end{enumerate}
As before, we shall give the technical details of the whole parallel
partitioning algorithm here.
\begin{defn}[Parallel Partitioning]
 Parallel partitioning a sorted batch $T$ of items around a sorted
batch $P$ of pivot items is done via the following procedure:
\begin{block}
\item Insert $\infty$ into $P$ (as the rightmost leaf).
\item Then \textbf{push $T$ down $P$} by essentially the same pipelined
splitting scheme as in \nameref{sub:par-filter} (\ref{sub:par-filter}):

\begin{enumerate}
\item In parallel place at each node $v$ of $P$:\setstretch{.91}

\begin{itemize}
\item DedicatedQueues $v\a{queue}[1..2]$.
\item Bool $v\a{qclear}[i]:=false$ for each $i\in[1..2]$.
\item Bool $v\a{frozen}:=false$.
\item Pointer $v\a{split}:=null$.
\item Barrier $v\a{fed}$.\quad{}// see  \ref{def:barrier}
\end{itemize}
\item Define \textbf{feeding} $B$ to $v\a{queue}[i]$ to be pushing $B$
onto $v\a{queue}[i]$ and then reactivating $v\a{pushdown}[i]$.
\item Define \textbf{closing} $v$ to be calling $v\a{fed}\a{notify}()$
and then reactivating both $v\a{pushdown}[1]$ and $v\a{pushdown}[2]$.
\item Start by feeding $T\a{root}$ to $P\a{root}\a{queue}[1]$ and then
closing $P\a{root}$.
\item \br Whenever $v\a{pushdown}[i]$ is reactivated for some node $v$
of $P$, it does the following:

\begin{block}
\item If $v$ is a leaf, return.
\item Set $done:=v\a{fed}\a{notified}()$.
\item Pop subtree $B$ off $v\a{queue}[i]$.
\item If $B=null$ (i.e.~$v\a{queue}[i]$ was empty):

\begin{block}
\item If $done$:\quad{}// if no more subtrees will be fed to $v$ before
$v\a{queue[i]}$ is found empty //

\begin{block}
\item Set $v\a{qclear}[i]:=true$.
\item If $v\a{qclear}[3-i]$ and $\f{TryLock}(v\a{frozen})$:

\begin{block}
\item If $v\a{split}\ne null$, call $v\a{split}\a{wait}()$.
\item Close $v\a{left}$ and close $v\a{right}$.
\end{block}
\end{block}
\item Return.
\end{block}
\item Reactivate $v\a{pushdown}[i]$.
\item If $B\a{last}\le v\a{left}\a{last}$, feed $B$ to $v\a{left}\a{queue}[i]$
and return.
\item If $B\a{first}>v\a{left}\a{last}$, feed $B$ to $v\a{right}\a{queue}[i]$
and return.
\item Set $v\a{split}:=\new{Barrier}$.\quad{}// see  \ref{def:barrier}
\item Fork the following \textbf{splitting process}:

\begin{block}
\item While $B$ is not a leaf:

\begin{block}
\item If $B\a{left}\a{last}\le v\a{left}\a{last}$:

\begin{block}
\item Feed $B\a{left}$ to $v\a{left}\a{queue}[2]$ and set $B:=B\a{right}$.
\end{block}
\item Otherwise:

\begin{block}
\item Feed $B\a{right}$ to $v\a{right}\a{queue}[1]$ and set $B:=B\a{left}$.
\end{block}
\end{block}
\item If $B\a{last}\le v\a{left}\a{last}$, feed $B$ to $v\a{left}\a{queue}[2]$,
otherwise feed $B$ to $v\a{right}\a{queue}[1]$.
\item Call $v\a{split}\a{notify}()$.
\end{block}
\end{block}
\end{enumerate}
\item And (after starting the push-down phase) \textbf{collate} the results
by doing the following for each leaf $v$ of $P$ in parallel:

\begin{enumerate}
\item Call $v\a{fed}\a{wait}()$.
\item Join the batches in each $v\a{queue}[i]$ (by repeatedly joining the
last two until there is only one left).
\item Tag $v$ with the join of the results.
\end{enumerate}
\end{block}
\end{defn}
The correctness of the push-down phase in this parallel partitioning
algorithm follows in the same way as for parallel filtering. To check
the correctness of the whole algorithm, we just need to observe the
following invariants:
\begin{enumerate}
\item If at any time $v\a{fed}\a{notified}()$ is $true$, then there will
be no more feeding to any queue of $v$.
\item If at any time $v\a{fed}\a{notified}()$ is $true$ and $v$ is an
internal node that has empty queues and no ongoing splitting process,
then eventually both $v\a{left}\a{fed}\a{notified}()$ and $v\a{right}\a{fed}\a{notified}()$
will be $true$ and remain $true$.
\end{enumerate}
These invariants imply that eventually $v\a{fed}\a{notified}()$ will
become and remain $true$ for every leaf of $P$, after which the
results will be collated.

We now prove two lemmas, which we shall not only use to prove the
cost bounds for parallel partitioning but also employ later for the
\nameref{sec:P23T} (\ref{sec:P23T}).
\begin{lem}[BBT Slice Joining]
\label{lem:bbt-slice-join} Any ordered slice $S$ of a BBT $T$
containing $k$ leaves can be joined into a single BBT in $O(\log(k+1))$
sequential time.\end{lem}
\begin{proof}
$S$ is the concatenation of two ordered slices such that in each
of them the subtrees have monotonic height with at most one pair of
the same height. Thus we can join the subtrees in each ordered slice
from shortest to tallest, taking $O(1)$ time per join, and then join
the two results in $O(\log(k+1))$ time.
\end{proof}

\begin{lem}[BBT Log Sum Bound]
\label{lem:bbt-log-sum} Take any real $n$, and any BBT $T$ with
$k$ leaves and a non-negative real weight-function $m$ on its nodes
such that $\sum_{v\in H(T,h)}m(v)\le n$ for every $h\in[0..T\a{height}]$.
Then $\sum_{v\in V(T)}\log(m(v)+1)\in O\left(k\cdot\log\left(\frac{n}{k}+1\right)+k\right)$.
\end{lem}
\begin{proof}
Let $c(h)$ be the size of $H(T,h)$. Each node in $H(T,h)$ has at
least $(\frac{3}{2})^{h}$ leaves by induction, and hence $c(h)\le k\cdot r^{h}$
where $r=\frac{2}{3}$. And $\sum_{v\in H(T,h)}\log(m(v)+1)\le c(h)\cdot\log\left(\frac{n}{c(h)}+1\right)$
by Jensen's inequality. Since $c\cdot\log\left(\frac{n}{c}+1\right)$
increases when $c$ increases, $\sum_{v\in V(T)}\log(m(v)+1)\le\sum_{h=0}^{\infty}\left(k\cdot r^{h}\cdot\log\left(\frac{n}{k\cdot r^{h}}+1\right)\right)\le\sum_{h=0}^{\infty}\left(k\cdot r^{h}\cdot\left(\log\left(\frac{n}{k}+1\right)+h\cdot\log\frac{1}{r}\right)\right)\in O\left(k\cdot\log\left(\frac{n}{k}+1\right)+k\right)$.
\end{proof}

\begin{thm}[Parallel Partitioning Costs]
\label{thm:par-part-cost} Parallel partitioning a sorted batch $T$
of $n$ items around a sorted batch $P$ of $k$ pivots takes $O\left(k\cdot\log\left(\frac{n}{k}+1\right)+k\right)$
work and $O(\log n+\log k)$ span.\end{thm}
\begin{proof}
First we bound the work taken by the push-down phase. Preparing $P$
(i.e.~inserting $\infty$ and initializing the nodes) takes $O(\log n+k)$
work and $O(\log k)$ span, and $\log n<\log(n+1)+(k-1)\cdot\log1\le k\cdot\log\left(\frac{n}{k}+1\right)$
by Jensen's inequality. Observe that the remaining work is $O(1)$
times the number of feedings plus $O(1)$ times the number of times
$v\a{frozen}$ is changed from $false$ to $true$ for some node $v$.
The latter is clearly $O(k)$, so it suffices to bound the number
of feedings.

The subtrees fed to each node $v$ of $P$ form a slice of $T$ (\ref{def:bt-slice}),
so the number of feedings to $v$ is at most $4\cdot\log_{2}(m(v)+1)$
where $m(v)$ is the total number of items in that slice (since $T$
is $4$-log-splitting). And clearly $\sum_{v\in H(P,h)}m(v)\le n$
for every $h\in[0..P\a{height}]$. Thus by the \nameref{lem:bbt-log-sum}
(\ref{lem:bbt-log-sum}) the total number of feedings is $O\left(k\cdot\log\left(\frac{n}{k}+1\right)+k\right)$.

Next we bound the span taken by the push-down phase. By the same proof
as for \nameref{thm:par-filter-cost} (\ref{thm:par-filter-cost}),
every run of $v\a{pushdown}[i]$ that processes some subtree of $T$
ends within $O(T\a{height}+P\a{height}+1)\wi O(\log n+\log k)$ span
after the push-down phase starts, since every such run must come before
any run of $v\a{pushdown}[i]$ that reaches ``$\f{TryLock}(v\a{frozen})$''
(at which point no more subtree will be processed at $v$). And every
splitting process ends within $O(\log n)$ span after the run that
forked it ends. And each node will be closed within $O(1)$ span after
all these runs and splitting processes have ended and its parent (if
any) has been closed. Therefore the push-down phase takes $O(\log n+\log k)$
span.

Finally, we bound the work and span taken by the collation phase.
The waiting clearly takes $O(k)$ work and $O(\log k)$ span. Joining
the subtrees at each leaf $v$ of $P$ takes $O\left(\log(m(v)+1)\right)\wi O(\log n)$
work/span (see \ref{lem:par-filter-inv} and \ref{lem:bbt-slice-join}).
Thus the collation phase takes $O(1)$ times the work taken by the
push-down phase, and we are done.
\end{proof}

\subsection{Parallel Joining}

\label{sub:par-join}

Incidentally, parallel joining (concatenating) of a batch of batches
can be done efficiently by using a simple parallel filtering.
\begin{defn}[Parallel Joining]
\label{def:par-join} Parallel joining a batch $T$ of $b$ batches
is done via the following 2-phase procedure:
\begin{enumerate}
\item In parallel replace each leaf $v$ of $T$ by the BBT for the batch
at $v$ if that batch is non-empty (i.e.~if the batch $B$ stored
at $v$ is stored in a BBT with root $r$, then copy the fields of
$r$ to $v$), otherwise put a dummy item ($null$) at $v$.
\item Parallel filter (\ref{sub:par-filter}) $T$ to remove dummy items
and hence obtain the desired output batch $U$.
\end{enumerate}
\end{defn}
\begin{thm}[Parallel Joining Costs]
\label{thm:par-join-cost} Parallel joining a batch of $b$ batches
with total size $n$ takes $O(b+n)$ work and $O(\log b+\log n)$
span.\end{thm}
\begin{proof}
Phase~1 clearly takes $O(b+n)$ work and $O(\log b+\log n)$ span,
after which $T$ is $4$-balanced since every BBT is $2$-balanced.
By the \nameref{thm:par-filter-cost-gen} (\ref{thm:par-filter-cost-gen}),
we are done.\end{proof}
\begin{rem*}
Alternatively, in phase~2 we can push $T$ down $U$ instead, after
which every non-dummy item of $T$ will have been pushed to a unique
leaf $v$ of $U$ and will be at one end of the slice of $T$ that
reaches $v$, and hence we can in $O(1)$ steps store that item at
$v$ as desired. $T$ may not be a BBT after phase~1, but still has
$O(\log b+\log n)$ height and is $8$-log-splitting. Thus the same
proof as for \nameref{thm:par-filter-cost} (\ref{thm:par-filter-cost})
works.
\end{rem*}
Note that if we have a batch of unsorted instances of $\tr$, rather
than just a batch of plain batches, then there is a more efficient
algorithm to join them (\ref{sub:P23T-join}).

\subsection{Parallel Merging}

\label{sub:par-merge}

Another useful operation is parallel merging of two sorted batches.
\begin{defn}[Parallel Merging]
 Parallel merging two sorted lists $A,B$ is done via the following
3-phase procedure:
\begin{enumerate}
\item Parallel partition (\ref{sub:par-part}) $B$ around $A$, resulting
in a part of $B$ at each leaf of $A$.
\item For each leaf $v$ of $A$ in parallel, insert the item of $A$ at
$v$ into the part of $B$ at $v$, optionally combining duplicates
of the same item. (This combining procedure can be any $O(1)$-time
procedure.)
\item Parallel join (\ref{sub:par-join}) the resulting batches at the leaves
of $v$.
\end{enumerate}
\end{defn}
\begin{thm}[Parallel Merging Costs]
\label{thm:par-merge-cost} Parallel merging sorted lists $A,B$
with total size $n$ takes $O(n)$ work and $O(\log n)$ span.\end{thm}
\begin{proof}
Let $k,m$ be the sizes of $A,B$ respectively. Then phase~1 takes
$O\left(k\cdot\log\left(\frac{m}{k}+1\right)+k\right)\wi O(k+m)$
work and $O(\log k+\log m)$ span (\ref{thm:par-part-cost}). Let
$c_{i}$ is the size of the part of $B$ that was at the $i$-th leaf
of $A$ after phase~1. Then phase~2 takes $O\left(\sum_{i=1}^{k}\log(c_{i}+1)+k\right)\wi O\left(k\cdot\log\left(\frac{m}{k}+1\right)+k\right)$
work (by Jensen's inequality) and $O(\log k+\log m)$ span. And phase~3
takes $O(n)$ work and $O(\log n)$ span (\ref{thm:par-join-cost}).\end{proof}
\begin{rem*}
Note that it is actually possible to merge two sorted lists of lengths
even more efficiently in some cases if they are stored in instances
of a specialized sorted-set data type, such as if they are stored
in parallel sorted-sets based on the batch-parallel 2-3 tree (\ref{sec:PSS})
and their sizes $k,m$ satisfy $k\ll m$.
\end{rem*}

\subsection{Parallel Sorting}

\label{sub:par-sort}

Now we come to the problem of sorting a batch of items from an arbitrary
set $S$ with a given (strict) linear ordering $<$ on $S$. As is
standard, let $S^{n}$ be the set of all length-$n$ sequences from
$S$. In our setting, a sorting algorithm is a procedure that given
any input batch of items from $S$ will output a batch containing
exactly the same items sorted in (non-strict) $<$-increasing order.
Each item in the input batch may have an arbitrary tag, and all copies
of an item are considered equal under the ordering regardless of their
tags, but these tags must be preserved in the output batch.

We shall first give the obvious parallel merge-sort $\f{PMSort}$
that we obtain from parallel merging, reproduced from \cite{OPFS}.
Incidentally, $\f{PMSort}$ outperforms the multithreaded sorting
algorithm in \cite{clrs2009algointro} (which takes $O(n\cdot\log n)$
work and $O\left((\log n)^{3}\right)$ span).
\begin{defn}[Parallel Merge-Sort]
\label{def:par-msort} Let $\f{PMSort}$ be the procedure that does
the following on an input batch $I$ of items:
\begin{block}
\item If $I\a{size}\le1$, return $I$. Otherwise, compute in parallel $A=\f{PMSort}(I\a{left})$
and $B=\f{PMSort}(I\a{right})$, and then parallel merge (\ref{sub:par-merge})
$A$ and $B$ into an item-sorted batch $C$, and then return $C$.
\end{block}
\end{defn}
\begin{thm}[$\f{PMSort}$ Costs]
 $\f{PMSort}(I)$ takes $O(n\cdot\log n)$ work and $O\left((\log n)^{2}\right)$
span for every sequence $I$ in $S^{n}$.\end{thm}
\begin{proof}
The claim follows directly from the work/span bounds for parallel
merging (\ref{thm:par-merge-cost}) and $I\a{height}\in O(\log n)$.
\end{proof}
$\f{PESort}$, also first given in \cite{OPFS}, is obtained from
$\f{PMSort}$ by simply tweaking it to combine multiple copies of
the same item into bundles, by combining bundles in each merging step.
$\f{PESort}$ that achieves the entropy bound for work but yet takes
only $O\left((\log n)^{2}\right)$ span on an input list of $n$ items.
\begin{defn}[Parallel Entropy-Sort]
\label{def:par-esort} Define a \textbf{bundle} of an item $x$ to
be a BT (binary tree) in which every leaf has a tagged copy of $x$.
Let $\f{PESort}$ be the parallel merge-sort variant that does the
following on an input batch $I$ of items:
\begin{block}
\item If $I\a{size}\le1$, return $I$. Otherwise, compute in parallel $A=\f{PESort}(I\a{left})$
and $B=\f{PESort}(I\a{right})$, and then parallel merge (\ref{sub:par-merge})
$A$ and $B$ into an item-sorted batch $C$ of bundles, combining
bundles of the same item into one by simply making them the child
subtrees of a new bundle, and then return $C$.
\end{block}
Then $\f{PESort}(I)$ returns an item-sorted batch of bundles, with
one bundle (of all the tagged copies) for each distinct item in $I$,
and clearly each bundle has height at most $I\a{height}$.\end{defn}
\begin{thm}[$\f{PESort}$ Costs]
\label{thm:par-esort-cost} $\f{PESort}(I)$ takes $O(H+n)\wi O(n\cdot\log u+n)$
work and $O\left((\log n)^{2}\right)$ span for every sequence $I$
in $S^{n}$ with item frequencies $q_{1..u}$, where $H=\sum_{i=1}^{u}\left(q_{i}\cdot\ln\frac{n}{q_{i}}\right)$.\end{thm}
\begin{proof}
The full proof can be found in \cite{OPFS}, but note that the key
lemma is simply that for each item $x$ with frequency $k$ in $I$
there are at most $O\left(k\cdot\log\frac{n}{k}+k\right)$ merging
steps whose result contains $x$, and that key lemma can be generalized
to \ref{lem:subtree-size-bound}.
\end{proof}
$\f{PESort}$ is not a full sorting algorithm, because $\f{PESort}(I)$
where $I$ is in $S^{n}$ returns a batch of bundles rather than a
batch of items. Nevertheless, this suffices for many applications
because each bundle in the output batch has height at most $I\a{height}\in O(\log n)$.
But we can in fact obtain a full parallel entropy-sorting algorithm.
Specifically, we can convert each bundle in $\f{PESort}(I)$ to a
batch (\ref{def:bundle-balance}), and then parallel join (\ref{sub:par-join})
all those batches to obtain the desired output.
\begin{defn}[Bundle Balancing]
\label{def:bundle-balance} A bundle $B$ of size $b$ and height
$h$ is balanced as follows:
\begin{block}
\item Recursively construct a linked list through all the leaves of $B$,
and mark the leaves of $B$ with ($1$-based) rank of the form $(i\cdot h+1)$,
and then extract those marked leaves as a batch $P$ by parallel filtering
(\ref{sub:par-filter}). Then at each leaf $v$ in $P$, construct
and store at $v$ a batch of the items in $B$ with ranks $i\cdot h+1$
to $(i+1)\cdot h$, obtained by traversing the linked list forward.
Now $P$ is essentially a batch of size-$h$ batches (except perhaps
the last smaller batch), which we then recursively join to obtain
the batch of all items in $G$.
\end{block}
\end{defn}

\begin{thm}[Bundle Balancing Costs]
 Balancing a bundle $B$ of size $b$ and height $h$ takes $O(b)$
work and $O(h)$ span.\end{thm}
\begin{proof}
Note that $B$ has less internal nodes than leaves, and so constructing
the linked list takes $O(b)$ work and $O(h)$ span. Extracting the
batch $P$ of items of $B$ with ranks at intervals of $h$ takes
$O\left(b+P\a{size}\cdot h\right)=O(b)$ work and $O(h)$ span, by
essentially the same proof as \ref{thm:par-filter-cost} since the
number of feedings is at most $P\a{size}\cdot h$. Constructing the
batches of items in-between those in $P$ takes $O(b)$ work and $O(P\a{height}+h)\wi O(h)$
span, and recursively joining them takes $O(1)$ work and span per
node of $P$ (except $O(h)$ span for the first joining involving
the last batch).
\end{proof}

\section{Batch-Parallel 2-3 Tree}

\label{sec:P23T}

We now present the batch-parallel 2-3 tree $\tr$ and explain how
to support the following operations on $\tr$:
\begin{itemize}
\item \textbf{Unsorted batch search/update:} Perform an unsorted batch of
$b$ searches/updates within $O(b\cdot\log n)$ work and $O(\log b\cdot\log n)$
span, tagging each search with the result and a direct pointer to
the item in $\tr$ (if it exists).
\item \textbf{Sorted batch access:} Perform an item-sorted batch of $b$
accesses (i.e.~searches/updates, inserts and deletes) to distinct
items within $O\left(b\cdot\log\left(\frac{n}{b}+1\right)+b\right)$
work and $O(\log b+\log n)$ span, tagging each access with the result
and a direct pointer to the item in $\tr$ (if it exists).
\item \textbf{Batch reverse-indexing:} Given an unsorted batch of $b$ direct
pointers to distinct items in $\tr$, return a sorted batch of those
items within $O\left(b\cdot\log\frac{n}{b}+b\right)$ work and $O(\log n)$
span.
\end{itemize}
Here $n$ is the current number of items in $\tr$, and a \textbf{direct
pointer} is an object that allows retrieving the item and its attached
value in $\tr$ at the time of the search in $O(1)$ steps. It must
also be used in the batch reverse-indexing operation. All these operations
must be performed sequentially (i.e.~one batch operation at a time).

Note that the sorted batch access requires that the accesses are to
distinct items, but there is no actual disadvantage to that constraint.
Suppose we are given an item-sorted input batch $I$ of $b$ accesses
that may have multiple accesses to be the same item. We can perform
an easy parallel recursion on $I$ to compute which accesses to an
item $x$ are the leftmost access to $x$ in $I$. Then we can recursively
join all the accesses to $x$ into a single batch $B_{x}$ (see \ref{lem:bbt-slice-join}),
store it at the leftmost access to $x$, and compute the effective
result of $B_{x}$ (if they are performed in order), within $O(1)$
work per access and $O(\log b)$ span. After that, we can parallel
filter (\ref{sub:par-filter}) out those leftmost accesses from $I$
to obtain an item-sorted batch $I'$ of the effective accesses, which
are to distinct items, within $O(b)$ work and $O(\log b)$ span (\ref{thm:par-filter-cost}).
We can now perform the usual sorted batch access on $I'$, and perform
one more parallel recursion to tag the original accesses in $I$ with
the appropriate results.

The sorted batch access also requires the input batch to be item-sorted.
But we can clearly also support unsorted batch access within $O(b\cdot\log\max(b,n))$
work and $O\left((\log b)^{2}+\log n\right)$ span, by using parallel
sorting (\ref{def:par-msort}) to sort the input batch. This is more
efficient than unsorted batch search/update when $1\ll b\ll n$. However,
unsorted batch search/update is useful when $b\gg n$, since sorting
the input batch would take more work. In fact, we shall see later
that if we judiciously combine both unsorted batch search/update and
sorted batch access with parallel entropy-sorting (\ref{def:par-esort}),
we can support \textbf{\textit{unsorted batch access}} within $O(b\cdot\log\max(n,n'))$
work and $O\left((\log b)^{2}+\log n\right)$ span where $n'$ is
the number of items in $\tr$ after that batch operation (\ref{sub:P23T-full-access}).

The batch reverse-indexing operation is useful in maintaining synced
instances of $\tr$ with the same items but sorted differently; we
can tag each item with direct pointers into other instances of $\tr$,
and after performing an $<$-sorted batch access on an $<$-sorted
instance $S_{1}$, we can use the obtained batch of direct pointers
into an $\vartriangleleft$-sorted instance $S_{2}$ to obtain the
corresponding $\vartriangleleft$-sorted batch of items, which we
can then use to perform the same accesses on $S_{2}$.

\subsection{Preliminaries \& Notation}

$\tr$ stores the items in a leaf-based 2-3 tree encoded as a leaf-based
red-black tree $T$ (i.e.~every red node has two black child nodes,
and every black node is either a leaf or has two child nodes at most
one of which is red, and the black nodes correspond to the nodes of
the 2-3 tree). From now on we shall drop the adjective ``leaf-based''
since we only use leaf-based trees.

For any 2-3 tree $X$ we shall also denote the children of a node
$v$ of $X$ by $v\a{\textbf{left}}$ and $v\a{\textbf{right}}$ and
$v\a{\textbf{mid}}$ (if it exists), and denote the height of $v$
in $X$ by $v\a{\textbf{height}}$. If $X$ is encoded as a red-black
tree $X'$ and $v$ corresponds to the node $v'$ in $X'$, then $v\a{left}$
would correspond to the first black descendant of $v'$ (and not necessarily
$v'\a{left}$), and likewise for $v\a{right}$, and $v\a{height}$
would be the number of black nodes excluding $v'$ in any path from
$v'$ to a leaf in $X'$. These apparent ambiguities will always be
resolved by the context, which will always specify whether we treat
a node as in a 2-3 tree or in a red-black tree.

For convenience, let $\boldsymbol{X+Y}$ denote the standard join
of 2-3 trees $X,Y$ in that order, and identify a 2-3 tree with its
root node. Also we shall write ``$\boldsymbol{X\sim Y}$'' and ``$\boldsymbol{X\gg Y}$''
as short for ``$X\a{height}=Y\a{height}$'' and ``$X\a{height}>Y\a{height}$''
respectively.

\subsection{Unsorted Batch Search}

\label{sub:P23T-unsorted}

Performing an \textbf{unsorted search/update} on an input batch $I$
of $b$ searches/updates is done by calling $\f{USearch}(T\a{root},I)$
(\ref{def:unsorted-search}). Note that we cannot simply spawn a thread
for each search that traverses $T$ from root to leaf, as it would
incur $\Omega(b)$ span at the root of $T$ in the queued contention
model (\ref{sub:prog-model}).
\begin{defn}[Unsorted Search/Update]
\label{def:unsorted-search}~
\begin{block}
\item \textbf{Private USearch( Node $v$ of BBT $T$ , Access Batch $B$
):}\quad{}// treat $T$ as a BBT

\begin{block}
\item If $B$ is empty, return.
\item If $v$ is a leaf:

\begin{block}
\item In parallel tag each access in $B$ to an item $x$ such that $x=v\a{last}$
with a direct pointer to $v$.
\item Check whether $B$ has an update, and if so execute the leftmost update
in $B$ on $v$.
\item Return.
\end{block}
\item Parallel partition $B$ around pivot $v\a{left}\a{last}$ into a lower
part $L$ and an upper part $R$ (see \ref{sub:par-filter}).
\item In parallel call $\f{USearch}(v\a{left},L)$ and $\f{USearch}(v\a{right},R)$
(and wait for both to finish).
\end{block}
\end{block}
\end{defn}
\begin{thm}[Unsorted Search Costs]
 $\f{USearch}(T\a{root},I)$ takes $O(b\cdot\log n)$ work and $O(\log b\cdot\log n)$
span.\end{thm}
\begin{proof}
Clearly we can ignore any call $\f{USearch}(v,B)$ where $B$ is empty.
Each call $\f{USearch}(v,B)$ with non-empty $B$ at an internal node
$v$ of $T$ takes $O(B\a{size})$ work and $O(\log b)$ span to partition
$B$ into $L$ and $R$ (\ref{thm:par-filter-cost}). Hence the entire
unsorted batch search takes $O(b\cdot\log n)$ work and $O(\log b\cdot\log n)$
span.
\end{proof}

\subsection{Sorted Batch Access}

\label{sub:P23T-normal}

Performing a \textbf{sorted access} on an item-sorted input batch
$I$ of accesses to distinct items is done in 3 phases:
\begin{enumerate}
\item \textbf{\uline{Splitting phase}}\textbf{:} Split the items in $T$
(treated as a BBT) around the items in $I\less I\a{last}$, using
\nameref{sub:par-part} (\ref{sub:par-part}) but without inserting
$\infty$ and without collating. Then every item $x$ in $T$ will
end up in some subtree of $T$ at a leaf $v$ of $I$ such that if
$x$ exists in $I$ then $x$ is at $v$.
\item \textbf{\uline{Execution phase}}\textbf{:} At each leaf $v$ of
$I$, join the subtrees of $T$ at $v$ into a single 2-3 tree (see
\ref{thm:23tree-slice-join}), and execute the access at $v$ is on
that 2-3 tree.
\item \textbf{\uline{Joining phase}}\textbf{:} Recursively join the 2-3
trees at the leaves of $I$ via a \textbf{\textit{pipelined joining
scheme}} that pushes those 2-3 trees down each other. Here is a high-level
overview and explanation of the algorithm:

\begin{enumerate}
\item Define the \textbf{spine structure} (\ref{def:spine-struct}) of a
non-root spine node $v$ (i.e.~along the leftmost/rightmost path)
of a 2-3 tree $X$ as the $v\a{height}$-bit binary number where the
$k$-th (most significant) bit is $1$ if the $k$-th spine node from
$v$ downwards (along the spine) has $3$ children, and is $0$ otherwise.
Then given any 2-3 trees and their left/right children's spine structures,
we can within $O(1)$ steps determine whether the join overflows (i.e.~is
taller than the original trees) and compute the left/right children's
spine structures for the join (\ref{thm:join-spine-struct}, and see
\ref{fig:join-spine} below).
\item Augmenting every 2-3 tree with spine structure (i.e.~the spine structure
of every non-root spine node $v$ is stored in $v\a{spine}$) allows
us to join any 2-3 tree $Y$ into $X$ \textbf{\textit{top-down}},
if $X\sim Y$ or $X\gg Y$, where we view the joining as pushing $Y$
down the spine of $X$, and at each node $v$ we perform a \textbf{\textit{local
adjustment}} that has the desired effect, based on $Y$ and $v\a{left},v\a{mid},v\a{right}$
and $v\a{left}\a{spine},v\a{right}\a{spine}$ alone. Specifically,
we immediately update $v\a{left},v\a{mid},v\a{right}$ and $v\a{left}\a{spine},v\a{right}\a{spine}$
to their final values after the join. This includes when $c+Y$ overflows
where $c$ is the next node along the way, in which case we also create
a blank child $w$ of $v$ and tag $Y$ with $w$, so that at $c$
we can move the overflowed subtrees to $w$ without having to access
$v$ again. (See \ref{tab:23tree-join} below for all needed local
adjustments.)
\item Observe that each local adjustment in the above top-down joining procedure
is \textbf{\textit{independent}} from other local adjustments made
at any other node in the 2-3 tree (\ref{lem:joinin-run-ind}), and
hence multiple join operations can be \textbf{\textit{pipelined}}
(\ref{lem:23tree-joined}), as long as the local adjustments at each
node remain in the same order.
\item This order constraint is easily achieved by using a dedicated queue
$v\a{queue}$ at each spine node $v$ of $X$ to maintain the 2-3
trees currently at $v$, and using a procedure $v\a{joinin}$ that
is run only via reactivation calls (see \ref{sub:prog-model}) to
process each 2-3 tree $X$ in $v\a{queue}$ one by one and perform
the appropriate local adjustment at $v$ before pushing $X$ down
to a child of $v$ if appropriate. To push a 2-3 tree down to a node
$w$, we push it onto (the back of) $w\a{queue}$ and then reactivate
$w\a{joinin}$. $v\a{joinin}$ also reactivates itself after it has
processed each 2-3 tree from $v\a{queue}$.
\item Putting everything above together: We just need to prepare each 2-3
tree $X$ at a leaf of $I$ by augmenting it with spine structure,
and then at each internal node of $I$ recursively compute the join
of the 2-3 trees computed by its children, pipelined in the above
manner. Then the root of $I$ would effectively compute the join of
all the 2-3 trees at the leaves of $I$, in the sense that its final
state after all queued trees have been processed is the desired join
(\ref{thm:23tree-join-correct}).
\item So at the end we just have to wait for all queued trees to be processed,
which can be done by waiting on a barrier $v\a{done}$ at every internal
node of $I$ (see \ref{sub:prog-model}), where $v\a{done}$ is notified
when the corresponding joining has finished. If that joining was of
$Y$ into $X$, it finishes after the local adjustment that makes
$Y$ a subtree of the resulting 2-3 tree, so we just have to tag $Y$
with $v\a{done}$ so that the local adjustment that finishes that
joining can notify $v\a{done}$.
\end{enumerate}
\end{enumerate}
It turns out that the same techniques used in the proof of the \nameref{thm:par-part-cost}
(\ref{thm:par-part-cost}) can be used to prove the desired work and
span bounds for the sorted batch access (\ref{thm:split-exec-guarantees},
\ref{thm:joining-work}, \ref{thm:joining-span}).

\begin{table}[H]
\begin{centerbox}
\begin{tabular}{|>{\raggedright}p{12em}|>{\raggedright}p{14em}|l|}
\hline 
\textbf{Operation} & \textbf{Case} & \textbf{Local Adjustment}\tabularnewline
\hline 
Join 2-3 trees $L$ and $R$

(where $L\sim R$ or $L\gg R$) & $L\sim R$ & \aligntop{\includegraphics[bb=0bp 0bp 79bp 37bp,clip]{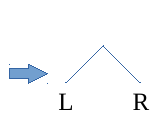}}\tabularnewline
\cline{2-3} 
 & $L\a{right}\sim R$  & \hspace{-1em}\aligntop{\includegraphics[bb=0bp 0bp 133bp 49bp,clip]{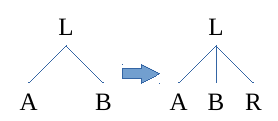}}\tabularnewline
\cline{3-3} 
 &  & \hspace{-1em}\aligntop{\includegraphics[bb=0bp 0bp 151bp 56bp,clip]{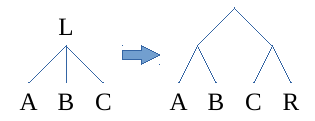}}\tabularnewline
\cline{2-3} 
 & $L\a{right}\gg R$

and $L\a{right}+R$ overflows & \hspace{-1em}\aligntop{\includegraphics[bb=0bp 0bp 194bp 49bp,clip]{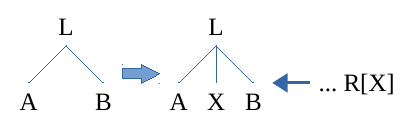}}\tabularnewline
\cline{3-3} 
 &  & \hspace{-1em}\aligntop{\includegraphics[bb=0bp 0bp 216bp 56bp,clip]{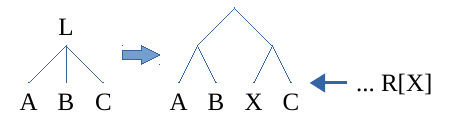}}\tabularnewline
\cline{2-3} 
 & $L\a{right}\gg R$

and $L\a{right}+R$ does not overflow & \aligntop{\includegraphics[bb=0bp 0bp 126bp 49bp,clip]{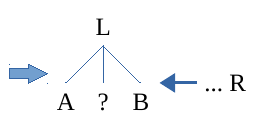}}\tabularnewline
\hline 
Join 2-3 tree $R$ to the right of 2-3 subtree $L$ (where $L\gg R$) & $L\a{right}\sim R$  & \hspace{-1em}\aligntop{\includegraphics[bb=0bp 0bp 155bp 49bp,clip]{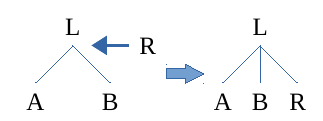}}\tabularnewline
\cline{3-3} 
 &  & \hspace{-1em}\aligntop{\includegraphics[bb=0bp 0bp 180bp 49bp,clip]{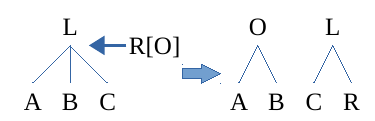}}\tabularnewline
\cline{2-3} 
 & $L\a{right}\gg R$

and $L\a{right}+R$ overflows & \hspace{-1em}\aligntop{\includegraphics[bb=0bp 0bp 216bp 49bp,clip]{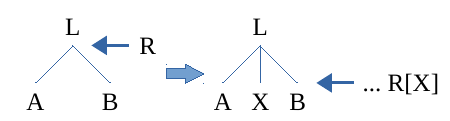}}\tabularnewline
\cline{3-3} 
 &  & \hspace{-1em}\aligntop{\includegraphics[bb=0bp 0bp 245bp 49bp,clip]{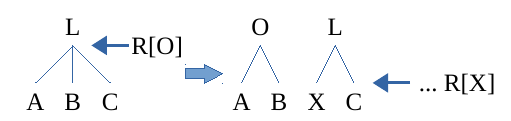}}\tabularnewline
\cline{2-3} 
 & $L\a{right}\gg R$

and $L\a{right}+R$ does not overflow & \hspace{-1em}\aligntop{\includegraphics[bb=0bp 0bp 202bp 49bp,clip]{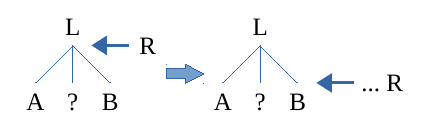}}\tabularnewline
\hline 
\end{tabular}

~

\textit{\small{}``$L\leftarrow R_{1}\ R_{2}\ ...$'' denotes that
$R_{1},R_{2},...$ are to be joined to the right of $L$ in that order.}\\
\textit{\small{}``?'' denotes that we do not care whether the (middle)
subtree exists (and we leave it as it is). }\\
\textit{\small{}``$X$'' denotes a newly created blank node, and
``$R[X]$'' denotes that $R$ is tagged with $X$.}{\small \par}
\end{centerbox}
\caption{\label{tab:23tree-join} Local adjustments for 2-3 tree joining on
the right (it is symmetrical on the left)}

\begin{centerbox}
~\end{centerbox}
\end{table}

\clearpage{}

We shall now fill in the technical details. First is the precise definition
of spine structure and the proof that it can be easily computed for
the result of any join without actually performing the join.
\begin{defn}[2-3 Tree Spine Structure]
\label{def:spine-struct} Take any 2-3 tree $X$. The \textbf{right
spine structure} of a node $v$ of $X$ is denoted by $\f{\textbf{rspine}}(v)$
and defined as $\sum_{k=1}^{v\a{height}}(c_{k}-2)\cdot2^{k-1}$ where
$c_{k}$ is the number of children of the node on the right spine
of the subtree at $v$ with distance $k$ from the leaf. Symmetrically
for the \textbf{left spine structure} of $v$ denoted by $\f{\textbf{lspine}}(v)$.
The \textbf{spine structure} of a non-root node $v$ of $X$ is denoted
by $\f{\textbf{spine}}(v)$ and defined as $\f{rspine}(v)$ if $v$
is a right child and $\f{lspine}(v)$ if $v$ is a left child. Note
that $\f{rspine}(v)$ is $0$ if $v$ is a leaf and is $\f{spine}(v\a{right})+(c-2)\cdot v\a{right}\a{weight}$
otherwise, where $c$ is the number of children of $v$ and $v\a{\textbf{weight}}=2^{v\a{height}}$,
and symmetrically for $\f{lspine}(v)$. Moreover, if $v$ is a right
child, then $\f{spine}(v)=\f{rspine}(v\a{parent})\%v\a{weight}$,
and symmetrically for a left child.
\end{defn}

\begin{thm}[2-3 Tree Join Spine Structure]
\label{thm:join-spine-struct} Take any 2-3 trees $X,Y$. Given $X\a{weight}$,
$\f{lspine}(X),\f{rspine}(X)$, $Y\a{weight}$, $\f{lspine}(Y),\f{rspine}(Y)$,
within $O(1)$ steps we can determine whether the join $J:=X+Y$ overflows
(i.e.~$J\gg X,Y$) and compute $\f{lspine}(J),\f{rspine}(J)$.\end{thm}
\begin{proof}
If $X\sim Y$, then $X+Y$ overflows, and $\f{lspine}(J)=\f{lspine}(X)$
and $\f{rspine}(J)=\f{rspine}(Y)$, so we are done. So by symmetry
we can assume $X\gg Y$. Clearly $X+Y$ overflows iff $\f{rspine}(X)+Y\a{weight}\ge X\a{weight}$.
During the joining of $Y$ to $X$, let $J'$ be the tree just after
adding $Y$ to $X$ as a right sibling of the right spine node $v$
of $X$ such that $v\sim Y$, and let $r$ be the root. Then $\f{rspine}(J')=\f{rspine}(X)-\f{rspine}(X)\%Y\a{weight}+Y\a{weight}+\f{rspine}(Y)$.
After that, whenever a non-root $4$-child node $w$ is split into
two $2$-child siblings, $\f{rspine}(r)$ is unchanged, since $w\a{parent}$
gains $1$ child while its rightmost child loses $2$ children, and
$2_{2}=10_{2}$. If $r$ is split, $\f{lspine}(r)$ decreases by $X\a{weight}$.
Thus $\f{rspine}(J)=\f{rspine}(J')\%X\a{weight}$. Finally, we can
check that $\f{lspine}(J)=\f{spine}(X\a{left})+(\f{rspine}(J)-\f{rspine}(J)\%J\a{right}\a{weight})$.
(See \ref{fig:join-spine} for an illustrated example.)
\end{proof}

\begin{figure}[H]
\begin{roundedboxinfloat}
\begin{centerbox}~

\includegraphics[bb=0bp 3bp 540bp 293bp,clip]{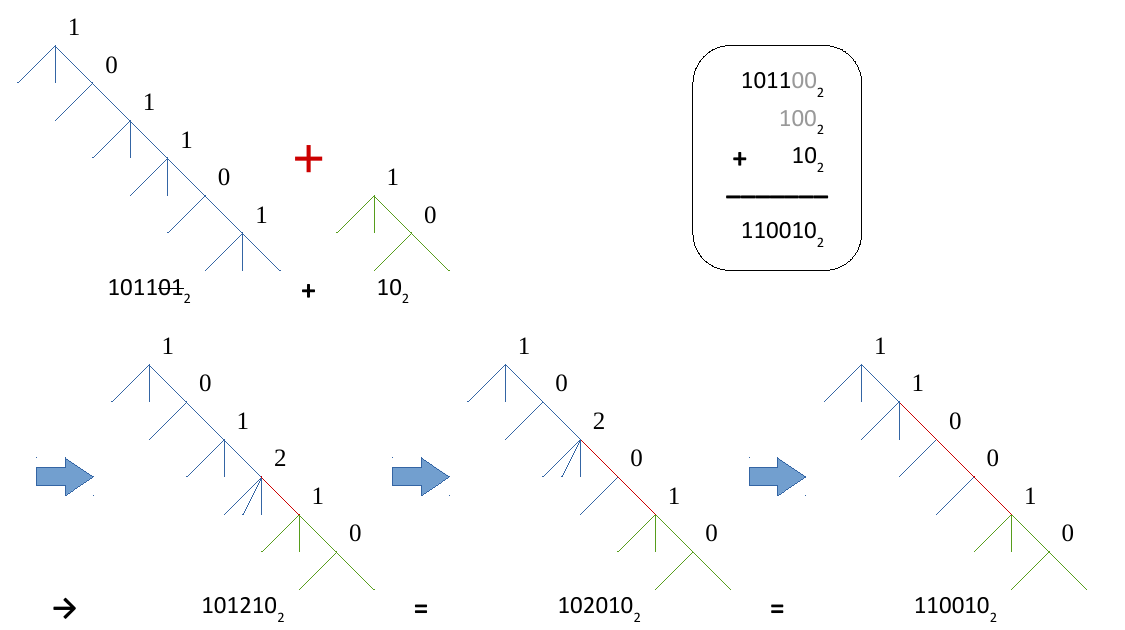}
\end{centerbox}
\caption{\label{fig:join-spine}Example computation of $\protect\f{rspine}(X+Y)$
for $X\gg Y$ within $O(1)$ steps given $X\protect\a{weight}$, $\protect\f{rspine}(X)$,
$Y\protect\a{weight}$ and $\protect\f{rspine}(Y)$. Here $X\protect\a{weight}={\color{magenta}1000000}_{2}$,
$Y\protect\a{weight}={\color{red}100}_{2}$, $\protect\f{rspine}(X)={\color{blue}1011}01_{2}$
and $\protect\f{rspine}(Y)={\color{green}10}_{2}$, so $X+Y$ does
not overflow since ${\color{blue}1011}01_{2}+{\color{red}100}_{2}<{\color{magenta}1000000}_{2}$,
and $\protect\f{rspine}(X+Y)={\color{blue}1011}00_{2}+{\color{red}100}_{2}+{\color{green}10}_{2}=110010_{2}$.}
\end{roundedboxinfloat}
\end{figure}

\clearpage{}

Next is the algorithm for executing an input batch $I$ of $b$ accesses,
which is done by calling $\f{Execute}(I)$ (\ref{def:sorted-access}).
\begin{defn}[Sorted Access]
\label{def:sorted-access}~
\begin{block}
\item Define \textbf{feeding} $B$ to $v\a{queue}[i]$ to be pushing $B$
onto $v\a{queue}[i]$ and then reactivating $v\a{pushdown}[i]$.
\item Define \textbf{closing} $v$ to be calling $v\a{fed}\a{notify}()$
and then reactivating both $v\a{pushdown}[1]$ and $v\a{pushdown}[2]$.
\item \textbf{Public Execute( Item-Sorted Access Batch $I$ ):}

\begin{block}
\item If $I$ is empty, return.
\item // Partition $T$ around $I\less I\a{last}$ by pushing $T$ down
$I$ //
\item For each node $v$ of $I$ in parallel:

\begin{itemize}
\item Create DedicatedQueues $v\a{queue}[1..2]$.
\item Create Procedure $v\a{pushdown}[i]$, which runs by calling $\f{PushDown}[i](v)$,
for each $i\in[1..2]$.
\item Initialize Bool $v\a{qclear}[i]:=false$ for each $i\in[1..2]$.
\item Initialize Bool $v\a{frozen}:=false$.
\item Initialize Pointer $v\a{split}:=null$.
\item Create Barrier $v\a{fed}$.\quad{}// see \ref{def:barrier}
\item Create Barrier $v\a{done}$.
\end{itemize}
\item Feed $T\a{root}$ to $I\a{root}\a{queue}[1]$ and then close $I\a{root}$.
\item // Rejoin the parts of $T$ after executing each access in $I$ on
the correct part //
\item Set $T:=\f{Collate}(I\a{root})$ and call $\f{Finalize}(I\a{root})$.
\end{block}
\item \textbf{Private PushDown$[i]$( Node $v$ of BBT $I$ ):}

\begin{block}
\item If $v$ is a leaf, return.
\item Set $done:=v\a{fed}\a{notified}()$.
\item Pop subtree $B$ off $v\a{queue}[i]$.\quad{}// treat $B$ as a BBT
\item If $B=null$ (i.e.~$v\a{queue}[i]$ was empty):

\begin{block}
\item If $done$:\quad{}// if no more subtrees will be fed to $v$ before
$v\a{queue[i]}$ is found empty //

\begin{block}
\item Set $v\a{qclear}[i]:=true$.
\item If $v\a{qclear}[3-i]$ and $\f{TryLock}(v\a{frozen})$:

\begin{block}
\item If $v\a{split}\ne null$, call $v\a{split}\a{wait}()$.
\item Close $v\a{left}$ and close $v\a{right}$.
\end{block}
\end{block}
\item Return.
\end{block}
\item Reactivate $v\a{pushdown}[i]$.
\item If $B\a{last}\le v\a{left}\a{last}$, feed $B$ to $v\a{left}\a{queue}[i]$
and return.
\item If $B\a{first}>v\a{left}\a{last}$, feed $B$ to $v\a{right}\a{queue}[i]$
and return.
\item Set $v\a{split}:=\new{Barrier}$.\quad{}// see  \ref{def:barrier}
\item Fork the following \textbf{splitting process}:

\begin{block}
\item While $B$ is not a leaf:

\begin{block}
\item If $B\a{left}\a{last}\le v\a{left}\a{last}$:

\begin{block}
\item Feed $B\a{left}$ to $v\a{left}\a{queue}[2]$ and set $B:=B\a{right}$.
\end{block}
\item Otherwise:

\begin{block}
\item Feed $B\a{right}$ to $v\a{right}\a{queue}[1]$ and set $B:=B\a{left}$.
\end{block}
\end{block}
\item If $B\a{last}\le v\a{left}\a{last}$, feed $B$ to $v\a{left}\a{queue}[2]$,
otherwise feed $B$ to $v\a{right}\a{queue}[1]$.
\item Call $v\a{split}\a{notify}()$.
\end{block}
\end{block}
\item \br \textbf{Private Collate( Node $v$ of BBT $I$ ):}

\begin{block}
\item If $v$ is a leaf:

\begin{block}
\item Call $v\a{fed}\a{wait}()$.
\item // Join all 2-3 trees at $v$ and execute the access at $v$ //
\item Convert each subtree in each queue $v\a{queue}[i]$ to a 2-3 tree
(i.e.~make its root node black).
\item Create empty 2-3 trees $L,R$.\quad{}// An empty 2-3 tree is stored
as $null$.
\item For each 2-3 tree $Q$ in $v\a{queue}[1]$ in reverse order, sequentially
join $Q$ onto the right of $L$.
\item For each 2-3 tree $Q$ in $v\a{queue}[2]$ in reverse order, sequentially
join $Q$ onto the left of $R$.
\item Sequentially join $L,R$ to obtain 2-3 tree $X$.
\item Execute the access at $v$ on $X$.
\item // Prepare the resulting 2-3 tree for pipelined joining //\quad{}
see \ref{def:23tree-join}
\item Call $\f{InitLeft}(w)$ for each non-root left spine node $w$ of
$X$ in order from leaf to root.
\item Call $\f{InitRight}(w)$ for each non-root right spine node $w$ of
$X$ in order from leaf to root.
\item Return $X$.
\end{block}
\item // Return the join of the results from the child nodes of $v$ //
\item In parallel set $L:=\f{Collate}(v\a{left})$ and $R:=\f{Collate}(v\a{right})$
(and wait for both to finish).
\item Return $\f{Join}(L,R,v\a{done})$.\quad{}// see \ref{def:23tree-join}
\end{block}
\item \textbf{Private Finalize( Node $v$ of BBT $I$ ):}

\begin{block}
\item If $v$ is a leaf, return.
\item In parallel call $\f{Finalize}(v\a{left})$ and $\f{Finalize}(v\a{right})$
(and wait for both to finish).
\item Call $v\a{done}\a{wait}()$.
\end{block}
\end{block}
\end{defn}
The subsequent procedures involve \textbf{23Trees}, where each 23Tree
$v$ is a 2-3 tree node in a weak sense; $v$ is either a leaf node
or an internal node with $2$ or $3$ child nodes that are 23Trees,
and $v\a{height}=w\a{height}+1$ for each child $w$ of $v$. 
\begin{defn}[23Tree Joining]
\label{def:23tree-join}~
\begin{block}
\item \textbf{Private LSpine( 23Tree $X$ ):}\quad{}// see \ref{def:spine-struct}

\begin{block}
\item Return $\cond{X\a{height}=0}0{X\a{left}\a{spine}+\cond{X\a{mid}=null}0{X\a{left}\a{weight}}}$.
\end{block}
\item \textbf{Private RSpine( 23Tree $X$ ):}\quad{}// see \ref{def:spine-struct}

\begin{block}
\item Return $\cond{X\a{height}=0}0{X\a{right}\a{spine}+\cond{X\a{mid}=null}0{X\a{right}\a{weight}}}$.
\end{block}
\item \textbf{Private InitLeft( 23Tree $X$ ):}\quad{}// prepares 2-3 tree
$X$ for joining as a left child

\begin{block}
\item Create Int $X\a{spine}:=\f{LSpine}(X)$.
\item Create DedicatedQueue $X\a{queue}$.
\item Create Procedure $X\a{joinin}:=\f{JoinLeft}(X)$.
\end{block}
\item \textbf{Private InitRight( 23Tree $X$ ):}\quad{}// prepares 2-3
tree $X$ for joining as a right child

\begin{block}
\item Create Int $X\a{spine}:=\f{RSpine}(X)$.
\item Create DedicatedQueue $X\a{queue}$.
\item Create Procedure $X\a{joinin}:=\f{JoinRight}(X)$.
\end{block}
\item \textbf{Private SJoin( 23Tree $L$ , 23Tree $R$ ):}\quad{}// returns
join of strict 2-3 subtrees $L,R$ with equal height

\begin{block}
\item Return new 23Tree $J$ with $J\a{left}:=L$ and $J\a{mid}:=null$
and $J\a{right}:=R$.
\end{block}
\item \textbf{Private RJoin( 23Tree $L$ , 23Tree $R$ ):}\quad{}// returns
join of root 2-3 trees $L,R$ with equal height

\begin{block}
\item Call $\f{InitLeft}(L)$ and $\f{InitRight}(R)$, and then return $\f{SJoin}(L,R)$.
\end{block}
\item \textbf{Private FeedLeft( 23Tree $X$ , 23Tree Node $v$ ):}\quad{}//
pushes $X$ to left child $v$ and updates the spine structure of
$v$ 

\begin{block}
\item Push $X$ onto $v\a{queue}$ and reactivate $v\a{joinin}$.
\item Set $v\a{spine}:=(v\a{spine}-v\a{spine}\%X\a{weight}+X\a{weight})\%v\a{weight}+\f{LSpine}(X)$.\quad{}//
see \ref{thm:join-spine-struct}
\end{block}
\item \textbf{Private FeedRight( 23Tree $X$ , 23Tree Node $v$ ):\quad{}}//
pushes $X$ to right child $v$ and updates the spine structure of
$v$

\begin{block}
\item Push $X$ onto $v\a{queue}$ and reactivate $v\a{joinin}$.
\item Set $v\a{spine}:=(v\a{spine}-v\a{spine}\%X\a{weight}+X\a{weight})\%v\a{weight}+\f{RSpine}(X)$.\quad{}//
see \ref{thm:join-spine-struct}
\end{block}
\item \br \textbf{Private Join( 23Tree $L$ , 23Tree $R$ , Barrier $done$
):}\quad{}// returns join of root 2-3 trees $L,R$; see \ref{tab:23tree-join}

\begin{block}
\item If $L=null$, call $done\a{notify}()$ and return $R$.
\item If $R=null$, call $done\a{notify}()$ and return $L$.
\item If $L\sim R$:

\begin{block}
\item Call $done\a{notify}()$ and return $\f{RJoin}(L,R)$.\hfill{}\overmid{\includegraphics{\string"P23T Diagrams/join1\string".pdf}}
\end{block}
\item If $L\gg R$:

\begin{block}
\item Call $\f{InitRight}(R)$.
\item If $L\a{right}\sim R$:

\begin{block}
\item Call $done\a{notify}()$.
\item If $L\a{mid}=null$:

\begin{block}
\item Set $L\a{mid}:=L\a{right}$ and $L\a{right}:=R$.\hfill{}\overmid{\includegraphics{\string"P23T Diagrams/join2\string".pdf}}
\item Return $L$.
\end{block}
\item Otherwise:

\begin{block}
\item Return $\f{RJoin}(\f{SJoin}(L\a{left},L\a{mid}),\f{SJoin}(L\a{right},R))$.\hfill{}\overmid{\includegraphics{\string"P23T Diagrams/join3\string".pdf}}
\end{block}
\end{block}
\item Otherwise:

\begin{block}
\item Set $R\a{joined}:=done$.
\item Create Pointer $R\a{overflow}:=null$.
\item If $L\a{right}\a{spine}+R\a{weight}\ge L\a{right}\a{weight}$:\quad{}//
$L\a{right}+R$ overflows; see \ref{thm:join-spine-struct}

\begin{block}
\item Create blank 23Tree (node) $X$ with same height as $L\a{right}$.
\item Set $R\a{overflow}:=X$.
\item $\f{FeedRight}(R,L\a{right})$.
\item If $L\a{mid}=null$:

\begin{block}
\item Set $L\a{mid}:=X$.\hfill{}\overmid{\includegraphics{\string"P23T Diagrams/join5\string".pdf}}
\item Return $L$.
\end{block}
\item Otherwise:

\begin{block}
\item Return $\f{RJoin}(\f{SJoin}(L\a{left},L\a{mid}),\f{SJoin}(X,L\a{right}))$.\hfill{}\overmid{\includegraphics{\string"P23T Diagrams/join4\string".pdf}}
\end{block}
\end{block}
\item Otherwise:

\begin{block}
\item $\f{FeedRight}(R,L\a{right})$.
\item Return $L$.\hfill{}\overmid{\includegraphics{\string"P23T Diagrams/join6\string".pdf}}
\end{block}
\end{block}
\end{block}
\item Symmetrically if $L\ll R$.
\end{block}
\item \textbf{Private JoinRight( 23Tree Node $L$ )} returns the Procedure
that runs as follows:

\begin{block}
\item Pop (the first) 23Tree $R$ off $L\a{queue}$.\hfill{}\textit{\textcolor{blue}{``$v\leftarrow...\ R$''
denotes that $R$ is pushed onto $v\a{queue}$}}
\item If $R=null$, return.\hfill{}\textit{\textcolor{blue}{``$R[X]$''
denotes that $R\a{overflow}=X$}}
\item If $L\a{right}\sim R$:

\begin{block}
\item If $L\a{mid}=null$:

\begin{block}
\item Set $L\a{mid}:=L\a{right}$ and $L\a{right}:=R$.\hfill{}\overmid{\includegraphics{\string"P23T Diagrams/joinr1\string".pdf}}
\end{block}
\item Otherwise:

\begin{block}
\item Copy (root of) $\f{SJoin}(L\a{left},L\a{mid})$ to $R\a{overflow}$.
\item Set $L\a{left}:=L\a{right}$ and $L\a{mid}:=null$ and $L\a{right}:=R$.\hfill{}\overmid{\includegraphics{\string"P23T Diagrams/joinr2\string".pdf}}
\end{block}
\item Call $R\a{joined}\a{notify}()$.
\end{block}
\item Otherwise:

\begin{block}
\item If $L\a{right}\a{spine}+R\a{weight}\ge L\a{right}\a{weight}$:\quad{}//
$L\a{right}+R$ overflows; see \ref{thm:join-spine-struct}

\begin{block}
\item Create blank 23Tree (node) $X$ with same height as $L\a{right}$.
\item If $L\a{mid}=null$:\hfill{}\overmid{\includegraphics{\string"P23T Diagrams/joinr3\string".pdf}}

\begin{block}
\item Set $L\a{mid}:=X$.
\end{block}
\item Otherwise:

\begin{block}
\item Copy (root of) $\f{SJoin}(L\a{left},L\a{mid})$ to $R\a{overflow}$.
\item Set $L\a{left}:=X$ and $L\a{mid}:=null$.\hfill{}\overup{\includegraphics{\string"P23T Diagrams/joinr4\string".pdf}}
\end{block}
\item Set $R\a{overflow}:=X$.
\end{block}
\item $\f{FeedRight}(R,L\a{right})$.\hfill{}\overmid{\includegraphics{\string"P23T Diagrams/joinr5\string".pdf}}
\end{block}
\item Reactivate $L\a{joinin}$.
\end{block}
\item \textbf{Private JoinLeft( 23Tree Node $R$ )} is symmetrically defined.
\end{block}
\end{defn}

\br

First we deal with the correctness and performance bounds for the
\textbf{\textit{splitting phase}} and \textbf{\textit{execution phase}}
(i.e.~the procedures Execute, PushDown, Collate in \ref{def:sorted-access},
excluding the preparation for pipelined joining in Collate).
\begin{thm}[Splitting+Execution Guarantees]
\label{thm:split-exec-guarantees} The splitting phase and execution
phase take $O\left(b\cdot\log\left(\frac{n}{b}+1\right)+b\right)$
work and $O(\log b+\log n)$ span, and their result is that the join
of the 2-3 trees at the leaves of $I$ is the desired final 2-3 tree.\end{thm}
\begin{proof}
The claims follow from \ref{thm:23tree-slice-join} below and the
same reasoning as for \nameref{sub:par-part} (\ref{sub:par-part}),
since they are essentially identical. In particular, the execution
phase takes $O(\log(c+1))$ work/span at each leaf $v$ of $I$ to
collate the 2-3 trees there with total size $c$, and then takes another
$O(\log(c+1))$ work/span to perform the access at $v$ on the result.\end{proof}
\begin{defn}[2-3 Tree Slice]
\label{def:23tree-slice} A \textbf{slice} of a 2-3 tree $T$ is
a minimal-length sequence of disjoint 2-3 subtrees of $T$ that contain
a set of consecutive leaves of $T$. An \textbf{ordered slice} of
$T$ is a slice of $T$ that has the subtrees listed in rightward
order in $T$.
\end{defn}

\begin{thm}[2-3 Tree Slice Joining]
\label{thm:23tree-slice-join} Any ordered slice $S$ of a 2-3 tree
$T$ containing $k$ leaves can be joined into a single 2-3 tree in
$O(\log k)$ sequential time.\end{thm}
\begin{proof}
$S$ is the concatenation of two ordered slices such that in each
of them the subtrees have monotonic height and there are at most two
subtrees of each height. Note that if 2-3 trees $X,Y,Z$ have height
in the range $[h-1,h]$, then $(X+Y)+Z$ is a 2-3 tree of height in
the range $[h,h+1]$ and takes $O(1)$ time. Thus we can join the
subtrees in each ordered slice from shortest to tallest, taking $O(1)$
time per join, and then join the two results in $O(\log k)$ time.
\end{proof}
So we will devote the rest of this section to analyzing the \textbf{\textit{joining
phase}} (i.e.~the pipelined joining preparation in Collate, and the
procedures SJoin, RJoin, FeedLeft, FeedRight, Join, JoinLeft, JoinRight
in \ref{def:23tree-join}, and Finalize in \ref{def:sorted-access}).
First we introduce some basic observations and terminology.

Observe that the 23Tree nodes involved always form a DAG (according
to the left/mid/right child pointers) at any point in time (since
children are always shorter). Based on this, we say that the \textbf{subtree
at} a 23Tree node $v$ is the set of nodes reachable by following
child pointers. We also say that a 23Tree $X$ is \textbf{queued at}
$v$ iff $X$ is in $v\a{queue}$, in which case we call $X$ a \textbf{queued
tree} at $v$. Additionally, we can impose a partial ordering on the
queued trees called the \textbf{pipeline order}, where a queued tree
$X$ at $v$ is \textbf{before} a queued tree $Y$ at $w$ iff $v$
is a strict descendant of $w$ or both $v=w$ and $X$ is before $Y$
in $v\a{queue}$.

Next observe that if 23Trees $L,R$ satisfy the property that $u\a{joinin}$
is defined for both its left and right child $u$, then $\f{Join}(L,R,\_)$
also does. And $v\a{joinin}$ is only defined via $\f{InitLeft}(v)$
or $\f{InitRight}(v)$, so we can check that it \textbf{feeds} to
a node $w$ (i.e.~calls $\f{FeedLeft}(\_,w)$ or $\f{FeedRight}(\_,w)$)
only if $w\a{joinin}$ is already defined. Hence the feedings done
by $v\a{joinin}$ are well-defined.

We shall call a run of $v\a{joinin}$ \textbf{effective} iff it pops
off a (non-$null$) tree $X$ from $v\a{queue}$, in which case we
say that it \textbf{processes} $X$. Note that each queued tree $X$
will be processed by $v\a{joinin}$ for each node $v$ that it is
fed to. Clearly ineffective runs have no effect, and we can from now
on view all the runs of $\f{Join},\f{JoinLeft},\f{JoinRight}$ as
\textbf{\textit{atomic}}, because runs of $\f{Join}$ clearly do not
interfere with each other, and because $v\a{joinin}$ is guarded by
the reactivation wrapper (\ref{def:reactivation}) and each effective
run of $v\a{joinin}$ is independent of other processes (\ref{lem:joinin-run-ind}).
\begin{lem}[Joinin Run Independence]
\label{lem:joinin-run-ind} For each 23Tree node $v$ involved in
the joining phase, what the sequence of all effective runs of $v\a{joinin}$
do is independent of any other runs of $\f{Join}$ or $\f{JoinLeft}$
or $\f{JoinRight}$. (In other words, the effect of those runs only
depends on the initial subtree at $v$ and the sequence of queued
trees processed by $v\a{joinin}$.)\end{lem}
\begin{proof}
Note that every 23Tree node $w$ has constant $w\a{height}$ (and
hence $w\a{weight}$) that was fixed at its creation. Thus each run
of $v\a{joinin}$ that processes a queued tree $X$ depends only on
$X$ and the fields $v\a{left},v\a{mid},v\a{right},v\a{left}\a{spine},v\a{right}\a{spine}$,
and so it suffices to show that these fields are modified only by
$v\a{joinin}$. Other runs of $\f{JoinLeft}$ or $\f{JoinRight}$
besides those of $v\a{joinin}$ will not modify these fields, since
they can only do so if $v=X\a{overflow}$ for some queued tree $X$,
but that is impossible because $X\a{overflow}$ is always a blank
23Tree node at the point when it is set, and neither $\f{InitLeft}$
nor $\f{InitRight}$ is ever called on it after that, so $X\a{overflow}\a{joinin}$
is never defined. And a run of $\f{Join}$ can only modify these fields
if it returns $v$, but clearly all runs of $v\a{joinin}$ can only
start after that run of $\f{Join}$ has returned.
\end{proof}
This independence lemma also implies that given any 23Tree $X,$ the
result of processing all queued subtrees in $X$ (without ever feeding
$X$) is uniquely determined by the current state of $X$, and we
can define that result to be the joined state of $X$, as made precise
in the next lemma.

\br
\begin{lem}[23Tree Joined State]
\label{lem:23tree-joined} Define a \textbf{joining sequence} for
a 23Tree $X$ to be a sequence of effective joinin runs on the subtree
at $X$ (i.e.~each is an effective run of $v\a{joinin}$ for some
node $v$ in the subtree at $X$) that processes queued trees in (the
subtree at) $X$ until there are none left. Then every joining sequence
for $X$ terminates and yields the same resulting subtree at $X$,
which we call the \textbf{joined state} of $X$.\end{lem}
\begin{proof}
Observe that each queued tree $Y$ in $X$ that is processed by $v\a{joinin}$
is either fed to a child of $v$ if $v\gg Y$ or stops being a queued
tree if $v\sim Y$. Thus $Y$ eventually stops being a queued tree,
and so every joining sequence eventually terminates. Now observe that
each effective run (viewed atomically) does not change the pipeline
order on the (remaining) queued trees, and hence all effective runs
of $v\a{joinin}$ at any particular node $v$ of $X$ process exactly
the same queued trees in exactly the same (pipeline) order regardless
of which joining sequence for $X$ is executed, hence yielding the
same result because of their independence (\ref{lem:joinin-run-ind}).
\end{proof}
Note that during the joining phase, the joined state of a node $v$
involved may change (if $v$ is fed).

We can now state and prove the correctness of $\f{Join}$ (\ref{thm:23tree-join-correct}),
and then the correctness of sorted accesses (\ref{thm:sorted-access-correct}).
\begin{defn}[2-3 Tree With Spine Structure]
\label{def:23tree-with-spine} We say that a 2-3 tree $X$ is \textbf{with
spine structure} iff every non-root spine node $v$ of $X$ has spine
structure (\ref{def:spine-struct}) $v\a{spine}$.\end{defn}
\begin{thm}[23Tree Joining Correctness]
\label{thm:23tree-join-correct} Take any run of $\f{Join}$ on 23Trees
$L,R$ that returns the 23Tree $J$ during the joining phase. Let
$L',R'$ be the joined states of $L,R$ respectively just before that
run, and let $J'$ be the joined state of $J$ just after that run.
If $L',R'$ are 2-3 trees with spine structure, then $J'$ is also
a 2-3 tree with spine structure, and furthermore $J'=L'+R'$ (i.e.~$J'$
is the standard join of $L'$ and $R'$).\end{thm}
\begin{proof}
Assume as given that $L',R'$ are 2-3 trees with spine structure.
By symmetry we can also assume that $L\sim R$ or $L\gg R$. Let $S$
be the global state just before the $\f{Join}(L,R,done)$ run. We
shall now consider two possible sequences of runs of $\f{Join},\f{JoinLeft},\f{JoinRight}$
(treating them as atomic) that can be executed starting from the same
state $S$.

By definition, $J'$ is the resulting state of $J$ upon performing
the following in order:
\begin{enumerate}
\item Execute $J:=\f{Join}(L,R,done)$ (which may make $R$ a queued tree
in $J$).
\item Process all queued trees in $J$.
\end{enumerate}
By \nameref{lem:joinin-run-ind} (\ref{lem:joinin-run-ind}), $J'$
is also the resulting state of $J$ upon performing the following
3 stages in order:
\begin{enumerate}
\item Process all queued trees in $L$.
\item Execute $J:=\f{Join}(L,R,done)$, and then process $R$ until it is
not a queued tree in $J$.
\item Process all (remaining) queued trees in $J$.
\end{enumerate}
Henceforth we shall assume this second sequence of runs. Clearly,
stage~1 makes $L$ become $L'$, and stage~3 makes $R$ (which is
a subtree on the right spine of $J$ after stage~2) become $R'$.
$L'$ is a 2-3 tree with spine structure, so stage~2 is effectively
performing the standard 2-3 tree joining algorithm to join $R$ on
the right of $L'$ (based on \ref{thm:join-spine-struct}), treating
$R$ as a 2-3 tree with height $R\a{height}$. Since we always have
$R\a{height}=R'\a{height}$, the result $J'$ of these 3 stages is
indeed $L'+R'$.

It remains to verify that $J'$ is with spine structure as well. The
non-root spine nodes of $J'$ comprise:
\begin{itemize}
\item The non-root left spine nodes of $L'$.
\item The non-root right spine nodes of $R'$.
\item The root of $R'$.
\item The non-root right spine nodes of $L'$ that $R$ was fed to (via
$\f{FeedRight}$).
\item The nodes $X,Y$ if the $\f{Join}(L,R,done)$ run calls $\f{RJoin}(X,Y)$.
\end{itemize}
First observe that $\f{Join}$ is run only on 23Trees that have no
queued tree at the root. Thus before the $\f{Join}(L,R,done)$ run,
$L$ has the same children as $L'$, each child $v$ of $L$ having
the same $v\a{spine}$ as in $L'$, and likewise for $R$. Moreover,
stages~2,3 do not feed any child of $R$, so $\f{RSpine}(R)=\f{rspine}(R')$
throughout all stages. Thus if $L\gg R$, then the $\f{Join}(L,R,done)$
run sets $R\a{spine}$ to $\f{RSpine}(R')$, and it is never changed
throughout stages~2,3.

Next observe that, on each call $C$ to $\f{FeedRight}(R,v)$ during
stage~2, if $v$ is a 2-3 tree with spine structure $v\a{spine}$
just before the call, then just after the call $v\a{spine}$ is the
spine structure $s$ of $v$ in $J'$. To see why, let $R_{0}$ be
the state of $R$ during $C$, and let $v_{0}$ and $v_{1}$ be the
state of $v$ just before and just after $C$ respectively, and consider
$v_{0}$ to also denote the subtree at $v_{0}$ just before $C$.
Then $v_{1}\a{spine}=(v_{0}\a{spine}-v_{0}\a{spine}\%R'\a{weight}+R'\a{weight})\%v_{0}\a{weight}+\f{rspine}(R')$,
because $R_{0}\a{weight}=R'\a{weight}$ and $\f{RSpine}(R_{0})=\f{rspine}(R')$
since $R_{0}\a{right}\a{spine}=R'\a{right}\a{spine}=\f{spine}(R'\a{right})$.
Also $\f{rspine}(v_{0}+R')=v_{0}\a{spine}-v_{0}\a{spine}\%R\a{weight}+R\a{weight}+\f{rspine}(R')$
by \ref{thm:join-spine-struct}. If $v_{0}+R'$ does not overflow,
then $v_{0}\a{spine}+R'\a{weight}<v_{0}\a{weight}$, and hence $v_{1}\a{spine}=\f{rspine}(v_{0}+R')=s$.
But if $v_{0}+R'$ does overflow, then $v_{1}\a{spine}=\f{rspine}(v_{0}+R')\%v_{0}\a{weight}=\f{spine}((v_{0}+R')\a{right})=s$.

With these observations, we can then check that if the $\f{Join}(L,R,done)$
run calls $\f{RJoin}(X,Y)$ on some nodes $X,Y$, then at that point
$X\a{left}\a{spine}$ and $Y\a{right}\a{spine}$ have been set to
$\f{spine}(X\a{left})$ and $\f{spine}(Y\a{right})$ respectively,
and hence $X\a{spine}$ and $Y\a{spine}$ will also be set to their
spine structure in $J'$.

Therefore we can now easily verify that every non-root spine node
$v$ in $J'$ has spine structure $v\a{spine}$.
\end{proof}

\begin{defn}[Finished 23 Tree]
 We say that a 23Tree $X$ is a \textbf{finished} iff it is a 2-3
tree with spine structure (\ref{def:23tree-with-spine}) and with
no queued trees (which implies that $X$ is its own joined state).\end{defn}
\begin{thm}[Sorted Access Correctness]
\label{thm:sorted-access-correct} Each call to $\f{Execute}(I)$
(i.e.~the sorted batch access on input batch $I$) eventually returns,
at which point $T$ is a finished 23Tree that matches the result of
performing $I$ on the original $T$ just before the call.\end{thm}
\begin{proof}
The call to $\f{Execute}(I)$ eventually returns, because each $\f{Join}(L,R,done)$
run either calls $done\a{notify}()$ or makes some 23Tree $X$ a queued
tree after setting $X\a{joined}:=done$, and whenever $X$ is processed
either it remains a queued tree or $X\a{joined}\a{notify}()$ is called,
so $\f{Finalize}(I\a{root})$ eventually returns. And the call to
$\f{Execute}(I)$ does not return until the joining phase is done,
because every queued tree $X$ is fed to a node by a unique $\f{Join}(L,R,done)$
run, before $\f{Finalize}()$ calls $done\a{wait}()$, and $X\a{joined}\a{notify}()$
is called only when $X$ is no longer a queued tree. The rest of the
claim follows from \ref{thm:23tree-join-correct}.
\end{proof}
Now we establish the work and span bounds for the \textbf{\textit{joining
phase}}.
\begin{thm}[Joining Phase Work]
\label{thm:joining-work} The joining phase takes $O\left(b\cdot\log\left(\frac{n}{b}+1\right)+b\right)$
work.\end{thm}
\begin{proof}
Observe that the work taken by the joining phase is $O(1)$ times
the total number of runs of $\f{Join}$ or $\f{JoinLeft}$ or $\f{JoinRight}$.
The number of runs of $\f{Join}$ is clearly at most $b-1$. The number
of runs of $\f{JoinLeft}$ or $\f{JoinRight}$ is at most twice the
number of feedings (via $\f{FeedLeft}$ or $\f{FeedRight}$), because
each reactivation of $\f{JoinLeft}(v)$ or $\f{JoinRight}(v)$ is
done either by $\f{FeedLeft}$ or $\f{FeedRight}$ or by itself, and
the number of self-reactivations is at most the number of queued trees
fed to $v$. So we shall bound the total number of feedings.

Each node $v$ of $I$ corresponds to a call $\f{Join}(L,R,done)$
for some 2-3 trees $L,R$. If $L\sim R$ then $\f{Join}(L,R,done)$
does not feed any queue. By symmetry it will suffice to analyze the
case that $L\gg R$. If $L\a{right}\sim R$ or $L+R$ overflows, then
$\f{Join}(L,R,done)$ also does not feed any queue. In the remaining
cases $R$ is pushed onto the queue at $L\a{right}$. Observe that
$R$ is pushed down at most $L\a{height}$ times before it is no longer
a queued tree. Thus the number of feedings of $R$ to a node is at
most $L\a{height}\le\log((L+R)\a{size}+1)=\log(m(v)+1)$ where $m(v)$
is the total number of items in all the 2-3 trees that were prepared
by $\f{Collate}(w)$ for some leaf $w$ of the subtree at $v$. And
clearly $\sum_{v\in H(I,h)}m(v)\le n$ for every $h\in[0..I\a{height}]$.
Thus by the \nameref{lem:bbt-log-sum} (\ref{lem:bbt-log-sum}) the
total number of feedings is $O\left(b\cdot\log\left(\frac{n}{b}+1\right)+b\right)$,
and we are done.\end{proof}
\begin{thm}[Joining Phase Span]
\label{thm:joining-span} The joining phase takes $O(\log b+\log n)$
span.\end{thm}
\begin{proof}
It is clear that $\f{Finalize}()$ takes $O(\log b+\log n)$ span
after all queued trees have been joined, because it only waits at
$v\a{done}$ at every node $v$ of the final tree $T$, whose height
is at most $O(\log n+I\a{height})\in O(\log b+\log n)$. Hence we
just have to show that all the joinin runs finish in $O(\log b+\log n)$
span after $\f{Collate}(I\a{root})$ finishes.

We shall use the same technique as in the proof of \nameref{thm:par-filter-cost}
(\ref{thm:par-filter-cost}). Call a run of $v\a{joinin}$ a \textbf{$v$-run},
and a reactivation of $v\a{joinin}$ a \textbf{$v$-reactivation}.
Call a $v$-run a \textbf{$(v,X)$-run} iff it processes the 23Tree
$X$ (i.e.~it pops $X$ off $v\a{queue}$). As before, each $v$-run
performs a $v$-reactivation iff it is effective, and there cannot
be three consecutive ineffective $v$-runs.

Also, each queued 23Tree $X$ during the joining phase is first pushed
onto a queue by some run of $\f{Join}$ that corresponds to some internal
node $u(X)$ of $I$. Let $d(X)$ be the depth of $u(X)$ in $I$,
and observe that $d(X')<d(X)$ for any $(v,X')$-run that precedes
a $(v,X)$-run. 

Now consider any $v$-run $R$. Let $X$ be the last 23Tree processed
by the $v$-runs up to $R$ (which exists since the first $v$-run
processes the first tree fed to $v\a{queue}$). There are two cases
(by \ref{thm:reactivation-prop} Property~3b):
\begin{itemize}
\item $R$ starts within $O(1)$ span after the end of the previous $v$-run
$R'$. Note that the last 23Tree $X'$ processed by the $v$-runs
up to $R'$ satisfies $d(X')\le d(X)$.
\item $R$ starts within $O(1)$ span after the start of the $v$-reactivation
$C$ that triggers $R$, and the reactivation point for $C$ is after
the end of any previous $v$-run. In this case, the previous $v$-run
$R'$ (if any) is ineffective, and hence $C$ must be executed by
the feeding of some 23Tree $X'$ to $v\a{queue}$. Observe that any
feeding of a 23Tree $X''$ to $v\a{queue}$ before $X'$ must push
$X''$ onto $v\a{queue}$ before executing some $v$-reactivation
$C'$ before $C$, and $C'$ must have reactivation point before the
start of $R'$ (since $C$ triggers $R$), and hence $X''$ must have
been processed by some $v$-run preceding $R$ since $R'$ is ineffective.
Thus $R$ would process $X'$ if it had not already been processed
by an earlier $v$-run, and hence $d(X')\le d(X)$.
\end{itemize}
From this and the fact that there cannot be three consecutive ineffective
$v$-runs, we can deduce that for every $(v,X)$-run $R$, at least
one of the following holds:
\begin{itemize}
\item $R$ starts within $O(1)$ span after the end of some preceding $(v,X')$-run
where $d(X')<d(X)$.
\item $R$ starts within $O(1)$ span after the start of some $v$-reactivation
executed by the feeding of some 23Tree $X'$ to $v\a{queue}$ where
$d(X')\le d(X)$.
\end{itemize}
Moreover, every feeding of a 23Tree $X$ to $v\a{queue}$ is executed
by either some $(v\a{parent},X)$-run or some run of $\f{Join}$.
Therefore by induction every $(v,X)$-run $R$ starts within $O(k+m+1)$
span after $\f{Collate}(I\a{root})$ finishes, where $k$ is the depth
of $v$ in $T$ and $m=d(X)$. Thus the whole joining phase finishes
within $O(T\a{height}+I\a{height}+1)\wi O(\log b+\log n)$ span after
$\f{Collate}(I\a{root})$ finishes.
\end{proof}

\subsection{Batch Reverse-Indexing}

\label{sub:P23T-reverse}

With the tools we have now, \textbf{reverse-indexing} is not too hard.
A \textbf{direct pointer} $X$ to an item in $\tr$ stores a private
pointer $X\a{node}$ to the leaf in $T$ that contains that item.
We augment each node $v$ of $T$ with $v\a{parent}$ storing its
parent node ($null$ if it is the root), and update it accordingly
during any of the other batch operations on $\tr$. We also augment
$v$ with a boolean flag $v\a{marked}$ initialized to $false$. Reverse-indexing
on an unsorted input batch $P$ of $b$ direct pointers to distinct
items is done in 2 phases (treating $T$ as a BBT throughout):
\begin{enumerate}
\item \textbf{\uline{Tracing phase}}\textbf{:} Recursively for each direct
pointer $X$ in $P$, traverse the path from $X\a{node}$ to the root,
where at each node $v$ the traversal is continued iff $\f{TryLock}(v\a{marked})$.
The spawning takes $O(b)$ work and $O(\log b)\wi O(\log n)$ span
since $P$ is a BBT and $b\le n$. Note that at most one traversal
will continue past each node and hence at most two traversals access
each flag. After all the traversals are done, every node $v$ along
the path from each desired leaf to the root of $T$ is marked (i.e.~$v\a{marked}=true$).
\item \textbf{\uline{Retrieving phase}}\textbf{:} Recursively traverse
$T$ top-down only through marked nodes, to find all the desired leaves
and join them into a 2-3 tree $U:=\f{Retrieve}(T\a{root})$ via the
same pipelined joining scheme as in the \nameref{sub:P23T-normal}
(\ref{sub:P23T-normal} joining phase). To wait for the joining to
be done, we wait on a barrier $v\a{done}$ at every marked node $v$
of $T$ with $2$ marked children, where $v\a{done}$ is notified
after the corresponding joining has finished. After all the joining
is done, recursively for each marked node $v$, call $\f{Unlock}(w\a{marked})$,
and call $v\a{done}\a{wait}()$ if $v$ has $2$ marked children.
After that, $U$ contains the desired items in sorted order, so return
$U$ converted to a batch.
\end{enumerate}

The same technique used to prove the \nameref{thm:joining-work} (\ref{thm:joining-work})
and \nameref{thm:joining-span} (\ref{thm:joining-span}) shows that
each joining of 2-3 trees during the retrieving phase, and the subsequent
waiting for the joining, takes $O(\log n)$ work and $O(\log n)$
span. Thus the total work taken is $O(b)$ plus $O(1)$ per marked
node plus $O(\log n)$ per joining, amounting to $O\left(b\cdot\log n\right)$.
But a more careful analysis will show that the joinings also take
only $O(1)$ work per marked node, and that there are only $O\left(b\cdot\log\frac{n}{b}+b\right)$
marked nodes, so in fact the total work is merely $O\left(b\cdot\log\frac{n}{b}+b\right)$.

The technical details are as follows. Reverse-indexing on an unsorted
input batch $P$ of direct pointers to distinct items is done by returning
$\f{ReverseIndex}(P)$.
\begin{defn}[Reverse-Indexing]
\label{def:reverse-index}~
\begin{block}
\item \textbf{Public ReverseIndex( Batch $P$ of DirectPointers into $\tr$
):}

\begin{block}
\item If $P$ is empty, return new empty Batch of items.
\item Call $\f{Trace}(X\a{node})$ for each DirectPointer $X$ in $P$ in
parallel (and wait for all to finish).
\item Create 23Tree $U:=\f{Retrieve}(T\a{root})$.
\item Call $\f{Finalize}(U\a{root})$.
\item Return $U$ converted to a Batch.\quad{}// easy since a 2-3 tree
is also a BBT
\end{block}
\item \textbf{Private Trace( Node $v$ of BBT $T$ ):}

\begin{block}
\item If $\f{TryLock}(v\a{marked})$ and $v\a{parent}\ne null$, call $\f{Trace}(v\a{parent})$.
\end{block}
\item \textbf{Private Retrieve( Node $v$ of BBT $T$ ):}

\begin{block}
\item If $v$ is a leaf, return new 23Tree containing only the item at $v$.
\item If $\neg v\a{right}\a{marked}$, return $\f{Retrieve}(v\a{left})$.
\item If $\neg v\a{left}\a{marked}$, return $\f{Retrieve}(v\a{right})$.
\item Create Barrier $v\a{done}$.\quad{}// see \ref{def:barrier}
\item In parallel set $L:=\f{Retrieve}(v\a{left})$ and $R:=\f{Retrieve}(v\a{right})$
(and wait for both to finish).
\item Return $\f{Join}(L,R,v\a{done})$.
\end{block}
\item \textbf{Private RFinalize( Node $v$ of BBT $T$ ):}

\begin{block}
\item $\f{Unlock}(v\a{marked})$.
\item If $v$ is a leaf, return.
\item If $\neg v\a{right}\a{marked}$, call $\f{RFinalize}(v\a{left})$
and return.
\item If $\neg v\a{left}\a{marked}$, call $\f{RFinalize}(v\a{right})$
and return.
\item In parallel call $\f{RFinalize}(v\a{left})$ and $\f{RFinalize}(v\a{right})$
(and wait for both to finish).
\item Call $v\a{done}\a{wait}()$.
\end{block}
\end{block}
\end{defn}
First we prove a simple lemma (\ref{lem:tracing-prop}) that will
be needed to prove the correctness and desired cost bounds of revese-indexing,
which is divided into the \textbf{\textit{tracing phase}} (i.e.~until
just before the call to $\f{Retrieve}(T\a{root})$) and the \textbf{\textit{retrieving
phase}} (i.e.~starting from the call to $\f{Retrieve}(T\a{root})$).
\begin{lem}[Tracing Properties]
\label{lem:tracing-prop} During the tracing phase, if there is a
call to $\f{Trace}(v)$, then the following properties hold:
\begin{enumerate}
\item Exactly one call to $\f{Trace}(v)$ evaluates $\f{TryLock}(v\a{marked})$
to $true$.
\item There is exactly one call from $\f{Trace}(v)$ to $\f{Trace}(v\a{parent})$
if $v$ is not the root.
\end{enumerate}
\end{lem}
\begin{proof}
Property~1 is obvious from the definition of $\f{TryLock}$ and the
fact that (during the tracing phase) no $v\a{marked}$ is ever set
to $false$. Property~2 is an immediate consequence of Property~1.
\end{proof}
From Property~1 of the foregoing lemma it is clear that, after the
tracing phase, every node $v$ of $T$ is marked (i.e.~$v\a{marked}=true$)
iff $v$ is along some path from a leaf $X\a{node}$ for some DirectPointer
$X$ in $P$. Thus the correctness of the retrieving phase follows
from the correctness of $\f{Join}$ (\ref{thm:23tree-join-correct})
in the same manner as the correctness of sorted batch access (\ref{thm:sorted-access-correct}).
So we are left with proving the desired cost bounds.

Next we bound the marked nodes using a combinatorial lemma, which
in fact applies not just to BBTs as needed here but also to any $O(1)$-log-splitting
full binary tree (see \ref{def:log-splitting}).
\begin{lem}[Subtree Size Bound]
\label{lem:subtree-size-bound} Take any $c$-log-splitting full
binary tree $T$, with $n$ leaves of which $k$ are marked where
$k>0$, and with each internal node marked iff it is on a path from
the root to a marked leaf. Then $T$ has at most $\left(c\cdot(k+1)\cdot\log_{2}\frac{n}{k}+2k\right)$
marked nodes.\end{lem}
\begin{proof}
Order the nodes of $T$ according to its in-order traversal. Let $v_{i}$
be the $i$-th marked leaf, and let $d_{i}$ be the number of unmarked
leaves (strictly) between $v_{i}$ and $v_{i+1}$. Also let $d_{0}$
be the number of unmarked leaves before $v_{1}$ and $d_{k}$ be the
number of unmarked leaves after $v_{k}$. Then there are at most $(c\cdot\log_{2}(d_{i}+1)+1)$
marked nodes between $v_{i}$ and $v_{i+1}$, since they are exactly
the least common ancestor of $v_{i},v_{i+1}$ plus the parent of each
subtree in the slice of $T$ that contains the leaves between $v_{i}$
and $v_{i+1}$, and that slice has $c\cdot\log_{2}(d_{i}+1)$ subtrees
(since $T$ is $c$-log-splitting). Similarly, there are at most $c\cdot\log_{2}(d_{0}+1)$
marked internal nodes before $v_{1}$ (since they are exactly the
parents of subtrees in the slice that contains the leaves before $v_{1}$),
and at most $c\cdot\log_{2}(d_{k}+1)$ marked internal nodes after
$v_{k}$. By Jensen's inequality we have $\sum_{i=0}^{k}\log_{2}(d_{i}+1)\le(k+1)\cdot\log_{2}\frac{n+1}{k+1}$.
Thus the total number of marked internal nodes is at most $c\cdot\sum_{i=0}^{k}\log_{2}(d_{i}+1)+(k-1)\le c\cdot(k+1)\cdot\log_{2}\frac{n+1}{k+1}+(k-1)<c\cdot(k+1)\cdot\log_{2}\frac{n}{k}+k$.
\end{proof}

\begin{thm}[Reverse-Indexing Costs]
 Reverse-indexing on an input batch of $b$ direct pointers to distinct
items takes $O\left(b\cdot\log\frac{n}{b}+b\right)$ work and $O(\log n)$
span.\end{thm}
\begin{proof}
Clearly spawning the calls to $\f{Trace}(X\a{node})$ for each DirectPointer
$X$ in $P$ takes $O(b)$ work and $O(\log b)\wi O(\log n)$ span.
And by \ref{lem:tracing-prop} Property~2, there are at most $2$
calls to $\f{Trace}(v)$ for each node $v$ in $T$, so $\f{TryLock}(v\a{marked})$
takes $O(1)$ time, and hence the calls to $\f{Trace}$ take $O(m)$
total work where $m$ is the number of marked nodes after the tracing
phase.

After that, the call to $\f{Retrieve}(T\a{root})$ takes $O(m)$ work
plus the work taken by the calls to $\f{Join}$, and takes $O(\log n)$
span. Note that $\f{Retrieve}(v)$ calls $\f{Join}$ only if $v$
has two marked children, and always returns a 23Tree with height at
most $v\a{height}$ by \nameref{thm:23tree-join-correct} (\ref{thm:23tree-join-correct})
and induction. Thus the pipelined processing triggered by each call
from $\f{Retrieve}(v)$ to $\f{Join}(L,R,v\a{done})$ (i.e.~pushing
$L$ down $R$ or vice versa) takes $O(\max(L\a{height},R\a{height})+1)\wi O(v\a{height})$
work by the same reasoning as in the proof of \nameref{thm:joining-work}
(\ref{thm:joining-work}), and hence the calls to $\f{Join}$ take
$O\left(\sum_{v\in C}v\a{height}\right)$ work in total, which is
$O(m)$ work by the \nameref{lem:tree-path-sum} (\ref{lem:tree-path-sum})
since $v\a{height}$ is at most twice the length of the shortest path
from $v$ to a leaf. Finally, $\f{RFinalize}(T\a{root})$ takes $O(m)$
work, and takes $O(\log n)$ span by the same reasoning as in the
proof of the \nameref{thm:joining-span} (\ref{thm:joining-span}).

Therefore the total work taken is $O(m)$, and $m\in O\left(b\cdot\log\frac{n}{b}+b\right)$
by the \nameref{lem:subtree-size-bound} (\ref{lem:subtree-size-bound}),
so we are done.
\end{proof}

\subsection{Unsorted Batch Access}

\label{sub:P23T-full-access}

We now explain how to implement \textbf{unsorted access} on an input
batch $B$ of $b$ accesses. How we perform $B$ depends on the current
size $n$ of $T$:
\begin{itemize}
\item If $b\le n$:

\begin{enumerate}
\item Parallel sort (\ref{def:par-msort}) $B$, taking $O(b\cdot\log b)$
work and $O\left((\log b)^{2}\right)$ span.
\item Use \nameref{sub:P23T-normal} (\ref{sub:P23T-normal}) to perform
the now item-sorted $B$ on $T$, taking $O(b\cdot\log n)$ work and
$O(\log b+\log n)$ span.
\end{enumerate}
\item If $b>n$:

\begin{enumerate}
\item Use \nameref{sub:P23T-unsorted} (\ref{sub:P23T-unsorted}) to perform
all search/updates in $B$ on $T$, including insertions on existing
items (which are treated as updates), and to perform all deletions
on non-existent items in $T$, taking $O(b\cdot\log n)$ work and
$O\left(\log b\cdot\log n\right)$ span.
\item Parallel filter (\ref{sub:par-filter}) out all the operations completed
in the previous step from $B$, taking $O(b)$ work and $O(\log b)$
span.
\item Parallel entropy-sort (\ref{def:par-esort}) the leftover batch $B$,
combining operations on the same item, taking $O\left(b\cdot\log m+b\right)$
work and $O\left((\log b)^{2}\right)$ span where $m$ is the final
size of $B$. Let $c$ be the final number of insertions in $B$,
and $d$ be the final number of deletions in $B$. Then the new size
$n'$ of $T$ is at least $c$, and $d\le n$ since there can only
be $n$ successful deletions. Thus $m=c+d\le n+n'$, and so $O\left(b\cdot\log m\right)\wi O(b\cdot\log(n+n'))\wi O(b\cdot\log\max(n,n'))$.
\item Use \nameref{sub:P23T-normal} (\ref{sub:P23T-normal}) to perform
the now item-sorted $B$ on $T$, taking $O(m\cdot\log n)$ work and
$O(\log m+\log n)$ span.
\end{enumerate}
\end{itemize}
In both cases, performing the batch $B$ takes $O\left(b\cdot\log\max(n,n')\right)$
work and $O\left(\left(\log b\right)^{2}+\log n\right)$ span, where
$n'$ is the size of $\tr$ after this batch operation.

\subsection{Batch Joining}

\label{sub:P23T-join}

We have finished describing how to implement batch operations on \textbf{\textit{sorted}}
instances of $\tr$, namely those whose items are in sorted order
(when listed according to the order of the leaves from left to right).
But we can also consider \textbf{\textit{unsorted}} instances of $\tr$,
namely those whose items are not required to be in sorted order. It
is easy to see that the \textbf{\textit{joining phase}} of the \nameref{sub:P23T-normal}
(\ref{sub:P23T-normal}) does not depend on the item ordering at all,
and so we can perform a \textbf{batch joining} of any batch $B$ of
$b$ unsorted instances of $\tr$ with total size $n$ by the same
\textbf{\textit{pipelined joining scheme}} (\ref{sub:P23T-normal}
joining phase). This takes $O\left(b\cdot\log\left(\frac{n}{b}+1\right)+b\right)$
work and $O(\log b+\log n)$ span, by the same proof as for \ref{thm:joining-work}
and \ref{thm:joining-span}.

\section{Optimal Parallel Sorted-Set}

\label{sec:PSS}

The red-black tree encoding the 2-3 tree in $\tr$ is a BBT, so it
can be used as an input batch on another instance of $\tr$. In particular,
given (sorted) instances $X,Y$ of $\tr$ with $m,n$ items respectively
such that $m\ge n$, to compute $X\cap Y$, $X\cup Y$ or $X\less Y$
we can treat $Y$ as an input batch (of searches, insertions or deletions
respectively) for $X$, taking $O\left(n\cdot\log\left(\frac{m}{n}+1\right)\right)$
work and $O(\log m+\log n)$ span. We can also compute $Y\less X=Y\less(X\cap Y)$
within the same bounds, since $i\cdot\log\left(\frac{n}{i}+1\right)<n+i$
where $i=\#(X\cap Y)\le n$. The work bound is \textbf{\textit{information-theoretically
optimal}} (in the comparison model), and the span bound is optimal
in the PPM model.

As written, the sorted batch insertions and deletions (\ref{sub:P23T-normal})
are destructive, and so $\tr$ is not persistent. However, if we want
to use $\tr$ only as a sorted-set data structure supporting intersections,
unions and difference, then we do not need the parent pointers used
in reverse-indexing (\ref{sub:P23T-reverse}), and hence it is not
hard to make $\tr$ \textbf{\textit{persistent}}. To do so, we modify
$\f{Collate}$ (\ref{def:sorted-access}) to perform the sequential
joins and execute the access non-destructively when computing the
resulting 2-3 tree $X$ at each leaf of the input batch $I$, and
then replace $X$ by a copy with just the left and right spine deep-copied.
This deep-copying is necessary because $\f{Join}$ (\ref{def:23tree-join})
may modify the spine nodes of $X$.

\section{Conclusions}

This paper presents a batch-parallel 2-3 tree for the QRMW PPM model
that is essentially optimal and can be trivially used to implement
an optimal sorted-set data structure. It can be seen that clever pipelining
can be used to achieve optimal work and span bounds that do not rely
on any `tricks' such as $O(1)$-time prefix-sums over all $p$ processors,
or $O(1)$ concurrent memory accesses over all $p$ processors, or
even pointer arithmetic. It raises the interesting question of just
how much can be done in the QRMW PPM model, which can be argued to
capture the intrinsic abstract costs of the problem.

Also, it is intriguing that the parallel sorted-set data structure
described in \cite{blelloch2019forkjoinalgo}, which is information-theoretically
optimal under a different computation model, namely the binary-forking
model with test-and-set, seems to crucially rely on concurrent reads
taking $O(1)$ time even if the contention is arbitrarily high, and
so cannot be translated to the QRMW model. However, that data structure
does not use any RMW operations besides test-and-set, whereas the
parallel sorted-set data structure in this paper (\ref{sec:PSS})
implicitly uses fetch-and-increment and fetch-and-decrement (in the
reactivation wrapper). Is such a trade-off necessary, or can we have
optimal parallel sorted-sets in the QRMW model that only use read,
write and test-and-set for memory accesses?

Lastly, this batch-parallel 2-3 tree supports sorted batch access
much more efficiently than unsorted batch access, and there seems
to be an intrinsic disparity. But is it possible to implement unsorted
batch access with better performance bounds?

\clearpage{}

\setlength{\baselineskip}{1em}

\phantomsection\bibliographystyle{plain}
\addcontentsline{toc}{section}{\refname}\bibliography{P23T}

\end{document}